\documentclass[a4paper, 12pt]{article}

\usepackage{amsfonts}
\usepackage{amssymb}
\usepackage{amsmath}

\usepackage{graphicx}
\usepackage[sort&compress,numbers]{natbib}
\usepackage{latexsym}
\usepackage{color}
\usepackage{braket}
\usepackage[colorlinks=true]{hyperref}
\usepackage{authblk}

\def\R{\Omega}
\usepackage{amsthm}

\usepackage[titletoc]{appendix}

\addcontentsline{toc}{section}{Appendices}

\newcommand{\bm}[1]{\mbox{\boldmath $#1$}}

\def\k{\kappa_\xi}

\def\be{\begin{equation}}
\def\ee{\end{equation}}
\def\bea{\begin{eqnarray}}
\def\eea{\end{eqnarray}}
\def\bean{\begin{eqnarray*}}
\def\eean{\end{eqnarray*}}

\def\K{K}
\def\L{L}
\def\={\stackrel{\Sigma}{=}}

\newlength{\cellwidth}
\usepackage{tabularx}
\usepackage{multirow}
\usepackage{array}

\newcolumntype{P}[1]{>{\centering\arraybackslash}p{#1}}
\newcolumntype{M}[1]{>{\centering\arraybackslash}m{#1}}

\newcommand{\one}[1]{{}^{(1)}#1}
\newcommand{\two}[1]{{}^{(2)}#1}

\def\torsion{\varsigma}

\def\Hxi{\H_{\xi}}

\def\Hhat{\widehat{\H}}

\def\Kill{{\mathcal A}}

\def\H{{\mathcal H}}

\def\la{\langle}
\def\ra{\rangle}

\newtheorem{theorem}{Theorem}
\newtheorem{definition}{Definition}
\newtheorem{remark}{Remark}
\newtheorem{corollary}{Corollary}
\newtheorem{proposition}{Proposition}
\newtheorem{lemma}{Lemma}

\def\Kill{{\mathcal A}}

\def\H{{\mathcal H}}

\def\Hhat{\widehat{\H}}

\def\la{\langle}
\def\ra{\rangle}

\newcounter{mnotecount}[section]

\renewcommand{\themnotecount}{\thesection.\arabic{mnotecount}}
\newcommand{\mnote}[1]
{\protect{\stepcounter{mnotecount}}$^{\mbox{\footnotesize
$
\bullet$\themnotecount}}$ \marginpar{
\raggedright\tiny\em
$\!\!\!\!\!\!\,\bullet$\themnotecount: #1} }

\newcommand{\mnotex}[1]
{\protect{\stepcounter{mnotecount}}$^{\mbox{\footnotesize $\bullet$\themnotecount}}$
\marginpar{
\raggedright\tiny\em
$\!\!\!\!\!\!\,\bullet$\themnotecount: #1} }

\newcommand{\ol}[1]{\overline{#1}{}}
\newcommand{\eq}[1]{(\ref{#1})}
\newcommand{\mcL}{{\pounds}}

\begin{document}

\title{Multiple Killing Horizons: the initial value formulation for $\Lambda$-vacuum}
\author[1]{Marc Mars}
\author[2]{Tim-Torben Paetz}
\author[3]{Jos\'e M. M. Senovilla}
\affil[1]{Instituto de F\'isica Fundamental y Matem\'aticas, Universidad de Salamanca, Plaza de la Merced s/n, 37008 Salamanca, Spain}
\affil[2]{Gravitational Physics, University of Vienna, Boltzmanngasse 5, 1090 Vienna, Austria}
\affil[3]{Departamento de F\'isica Te\'orica e Historia de la Ciencia, Universidad del Pa\'is Vasco UPV/EHU, Apartado 644, 48080 Bilbao, Spain}

\maketitle

\vspace{-0.2em}

\begin{abstract}
In \cite{mps} we have introduced the notion of  ``Multiple Killing Horizon'' and analyzed some of its general properties. Multiple Killing Horizons are Killing horizons for two or more linearly independent Killing vectors simultaneously.
In this paper we  focus on the vacuum case, possibly with cosmological constant, and study the emergence of Multiple Killing Horizons
in terms of characteristic initial value problems for two transversally intersecting null hypersurfaces. As a relevant outcome, a more general definition of Near Horizon Geometry is put forward. This new definition avoids the use of Gaussian null coordinates associated to the corresponding degenerate Killing vector and thereby allows for inclusion of its fixed points.
\end{abstract}

\section{Introduction}

In \cite{mps}  we have introduced \emph{Multiple Killing Horizons (MKHs)}: these are null hypersurfaces (or portions thereof) which are
 simultaneously the Killing horizons of two or more independent Killing vectors.  The {\em order} of a MKH is the number of linearly independent such Killing vectors.
In \cite{mps} we focused on basic concepts and properties of MKHs. In particular we derived
an equation which we called \emph{master equation}, and which is satisfied by the proportionality function between
two  Killing vectors on their common horizon.
We showed that there are two distinct types of MKHs, fully degenerate or not, the former having all surface gravities vanishing.

This paper is devoted to an analysis of the occurrence of MKHs if field equations are imposed. Specifically, we shall impose the Einstein's vacuum field
equations, possibly with a cosmological constant $\Lambda$ of any sign.
As Killing horizons are null hypersurfaces, $\Lambda$-vacuum spacetimes with MKHs can be generated in terms of a characteristic initial value problem.
While the analysis in \cite{mps} led, in form of the master equation, to necessary conditions for the existence of MKHs, this approach permits the derivation of necessary and sufficient
conditions on the characteristic data to generate a vacuum spacetime with a MKH.

There is a strong relationship between MKHs and near-horizon geometries, as discussed in \cite{mps2}. To put things in perspective, in \cite{LRS} the relation between near-horizon geometries and stationary black-hole holographs has been studied.
In this context the master equation is established  as a necessary and sufficient condition on the bifurcation surface of a bifurcate horizon
to be a non-degenerate MKH. Our purpose is to study the characteristic initial-value problem systematically and to analyze in
detail the emergence of Killing vectors and MKHs, including fully degenerate ones, as well as their order.

The paper is organized as follows:
In Section~\ref{sec2} we recall the most important definitions and results from \cite{mps,mps2}.
In Section~\ref{sec4} we study how $\Lambda$-vacuum spacetimes with MKHs arise via a characteristic initial value problem.
More specifically, in Section~\ref{sec_initial} we recall the characteristic initial value problem for two transversally intersecting
null hypersurfaces \`a la Rendall \cite{rendall}.
In Section~\ref{sec_KIDs} we recall the Killing Initial Data (KID) equations for this type of initial value problem \cite{KIDs}.
The KID equations are analyzed in Section~\ref{sec_bifurcate} for characteristic data which generate a vacuum spacetime
which admits a bifurcate Killing horizon.

This provides the basis to analyze  in Section~\ref{sec_non-deg} the emergence of bifurcate horizons which contain a (necessarily non-fully degenerate) MKH.
It turns out that the master equation considered on
the bifurcation surface
provides a necessary and sufficient criterion in this setting.
The remaining of the section contains several particular situations of relevance where non-degenerate MKHs arise. In particular,
in Section~\ref{sec_non_deg3} we focus on the case where the non-degenerate MKH is at least of order 3. It turns out that
in that case the spacetime needs to have a couple of additional Killing vectors.
In Section~\ref{sec_vanish_torsion} we analyze the emergence of MKH in  the case where the torsion one-form vanishes on the bifurcation surface (see section \ref{sec4} for definitions).
In particular we end up with data for Minkowski and (Anti-)de Sitter spacetime, respectively, if, in addition, the birfurcation surface is maximally symmetric.

We will show in Section~\ref{sec_non_deg_3+1} that, in $3+1$-dimensions, vacuum spacetimes with a bifurcate horizon which admits a MKH of order 3 do not exist, while examples where it is of order 4  are provided by maximally symmetric spacetimes.
For $\Lambda \neq 0$ the (Anti-) Nariai spacetime provides an example where the bifurcate horizon is a MKH of order 2.
In Section~\ref{sec_order2MKH} we construct, in arbitrary dimensions, data which generate a $\Lambda=0$-vacuum spacetime with a bifurcate horizon which
contains a MKH of order 2.

Section~\ref{sec_NHG} is devoted to the study of near-horizon geometries. In \cite{mps2} we have shown that the near-horizon geometry of a MKH is {\em unique}: it does
not depend on the Killing vector used to compute it. Here we provide a simpler proof of this result in the $\Lambda$-vacuum case, and discuss how the near-horizon limit
is performed  from the viewpoint of a characteristic initial-value  problem. This allows for an improved definition of near-horizon geometries which can deal with fixed points of the associated Killing vector.

Finally, in Section~\ref{sec_fully_deg} we construct a family of  characteristic initial data which generate vacuum spacetimes with fully degenerate horizons.
In turns out that an analysis of the KID equations is more involved in this setting.
In Section~\ref{sec_fully_deg} we will focus on the results while we have shifted most of the calculations to an Appendix.

\subsection{Notation}
$(M,g)$ denotes a connected, oriented and time-oriented $(n+1)$-dimensional Lorentzian manifold with metric $g$ of signature $(-,+,\dots,+)$.
$\overline{A}$ is the topological closure of a set $A$.
For any vector (field) $v$ in $TM$, $\bm{v}$ denotes the metrically
related one-form. Both index-free and index notation are used. Lowercase Greek letters $\alpha,\beta,\dots $ run from $0$ to $n$, small Latin indices $i,j,\dots$ take values from $1$ to $n$, $a\in\{1,2\}$, and capital Latin indices $A,B,\dots$ are running from $2$ to $n$, unless otherwise stated.

\section{Multiple Killing horizons}
\label{sec2}

The purpose of this section is to recall the most relevant definitions and results for spacetimes
which admit MKHs.
For more details we refer the reader to  \cite{mps}.

\begin{definition}\label{defi}
\begin{enumerate}
\item[(a)]
A smooth null hypersurface $\H_{\xi}$ embedded  in  a spacetime $(M,g)$ is
a {\bf Killing horizon  of a Killing vector $\xi$} of $(M,g)$
if and only if $\xi$ is null on $\H_\xi$, nowhere zero on $\H_\xi$ and tangent to $\H_\xi$.
We require that the interior of the  closure of $\H_{\xi}$  is a smooth connected hypersurface.
\item[(b)]
Let $\xi$ be a Killing vector on $(M,g)$ which has a connected and spacelike co-dimension two submanifold $S$ of fixed points. Then, the
set of points along all null geodesics orthogonal to $S$ forms  a {\bf bifurcate Killing horizon} with respect to $\xi$.
\end{enumerate}
\end{definition}
These null geodesics orthogonal to $S$ generate two transversal null hypersurfaces $\H_1$ and $\H_2$ whose future portions $\H_1^+$ and $\H_2^+$ (as well as its past portions $\H_1^-$ and $\H_2^-$) are connected Killing horizons. Notice that $\H_1^+ \cup \H_1^- \subset \H_1$ and $\H_2^+ \cup \H_2^-$ are Killing horizons according
to our definition. The union $\H_1^+\cup \H_2^+\cup\H_1^-\cup \H_2^-\cup S = \H_1\cup \H_2$ is the bifurcate Killing horizon.

Recall that the {\bf surface gravity} of a Killing horizon $\H_\xi$ is the function $\k$ defined by $$\nabla_{\xi}\xi|_{\H_{\xi}}= \kappa_{\xi}\xi .$$

\begin{definition}[\cite{mps}]
\label{thm_kappa_const}
A null hypersurface $\H$ embedded  in a spacetime $(M,g)$ is
a {\bf multiple Killing horizon (MKH)} if
 $(M,g)$ admits
Killing horizons  $\H_{\xi_i}$, $i\in\{1,\dots ,m\}$ with $m\geq 2$,
associated to linearly independent Killing vectors
$\xi_i$ satisfying
\begin{align*}
\overline{\H} = \overline{\H}_{\xi_1} = \dots = \overline{\H}_{\xi_m}.
\end{align*}
\end{definition}

\begin{theorem}[\cite{mps}]
All surface gravities of a multiple Killing horizon $\H$ are necessarily
constant on the entire $\H$.
\end{theorem}

Let  $\H$ be a MKH and define
$\Kill_\H$
as the  union of the trivial Killing
vector  and the collection of Killing vectors $\xi$ which admit
a Killing horizon $\H_{\xi}$ satisfying $\overline{\H} = \overline{\H_{\xi}}$.

\begin{theorem}[\cite{mps}]
\label{thm_MKH2}
Let $\H$ be a MKH  in an $(n+1)$-dimensional  spacetime
$(M,g)$ of  dimension at least two. Then:
\begin{enumerate}
\item[(i)] $\Kill_{\H}$ is a Lie sub-algebra of the Lie algebra of all Killing vectors  of $(M,g)$.
Its dimension $m\geq 2$ defines the {\bf order of the MKH}.
\item[(ii)]
 $\Kill_\H$  contains an Abelian sub-algebra $\Kill_\H^{deg}$ of dimension at least $m-1$ whose elements have vanishing surface gravities. If  $\Kill_\H^{deg}$ has dimension $m-1$,  any element $\xi\in \Kill_\H\setminus \Kill_\H^{deg}$ has $\k\neq 0$ and satisfies
$\left[\xi,\eta \right] =-\k \eta$, $\forall \eta \in \Kill_\H^{deg}$.
\item[(iii)] The maximum possible dimension of $\Kill_\H^{deg}$ is $n=$ {\em dim}$(M)-1$.
\end{enumerate}
\end{theorem}

\begin{definition}[\cite{mps}]
A MKH $\H$ is said to be {\bf fully degenerate} if $\Kill_\H = \Kill_\H^{deg}$, otherwise
it is called non-fully degenerate, or in short {\bf non-degenerate}.
\end{definition}
By Theorem~\ref{thm_MKH2} the maximum possible order of a MKH $\H$ is
\begin{enumerate}
\item[(i)] $m=n$ for  fully degenerate MKHs,
\item[(ii)] $m=n+1$ for non-degenerate MKHs.
\end{enumerate}

\begin{lemma}[\cite{mps}]
  \label{lemmamps}
Let $\H$ be a MKH, and let further  $\eta\in \Kill^{deg}_\H$ and $\xi\in \Kill_\H$, so that $\kappa_\eta =0$ and $\k$ may be zero or not.
Then  on $\Hxi$
there exist smooth functions $\tau, f_{\eta}: \Hxi \mapsto \mathbb{R}$ with $\xi(\tau)=1$ and $\xi(f_{\eta})=0$
such that  the following relation holds
\begin{equation}
\eta |_{\Hxi} = f_\eta e^{-\k \tau} \xi
\,.
\label{dfn_eqn_f}
\end{equation}

\end{lemma}

The level sets of the function $\tau$ define a foliation $\{S_{\tau}\}$ of
$\Hxi$ by spacelike co-dimension 2 surfaces.
Of course, there is a freedom, namely to  apply shifts $\tau \rightarrow \tau+\tau_0$ with $\xi(\tau_0)=0$, which induces a change as $f_\eta \rightarrow f_\eta e^{\k \tau_0}$. This freedom amounts to a change of the chosen foliation. However, for a Killing horizon its inherited first fundamental form
$\gamma$ satisfies $\pounds_\xi \gamma =0$, and thus all possible spacelike cross sections are isometric to each other with positive-definite metric $\gamma_{AB}$ (we keep the same name for simplicity).

Pick up any particular cross section $S_0\subset \Hhat$, not necessarily belonging to the chosen foliation $\{S_{\tau}\}$. Let $D$ be the canonical covariant derivative on $(S_0,\gamma)$ and $\{e_A\}$ a basis of vector fields on $S_0$. Due to the fact that Killing horizons are totally geodesic null hypersurfaces, the following relation holds
\be
D_{e_A}e_B -\nabla_{e_A} e_B = - {\cal K}_{AB}\,  \xi \label{2ndff}
\ee
which defines the {\em unique} non-vanishing second fundamental form ${\cal K}_{AB}$ of $(S_0,\gamma)$. Then, the {\bf torsion one-form} $\bm{s}$ on $(S_0,\gamma)$ relative to the chosen Killing $\xi$ is defined in the given basis by
\be
s_A \xi^\mu = -e^\rho_A\nabla_\rho \xi^\mu .\label{torsion}
\ee

It is clear from (\ref{dfn_eqn_f}) that the function $f: S_0 \rightarrow
\mathbb{R}$  defined by
$$f := f_\eta e^{-\k\tau} |_{S_0}$$
provides the proportionality between $\eta$ and $\xi$ on $S_0$: $\eta \stackrel{S_0}{=} f \xi$. Then we have (we correct an unfortunate typo in
\cite[expression (60)]{mps})
%

\begin{proposition} [\cite{mps}]
Let $\H$ be a MKH and  $\H_{\xi}$, $\H_{\eta}$
be Killing horizons satisfying $\overline{\H}_{\xi} =
\overline{\H}_{\eta} = \overline{\H}$.
Then the  following {\bf master equation} holds on any spacelike cut $S_{0}$ of $\Hxi$ 
%
\be
D_A D_B f - 2 s_{(A} D_{B)} f  +  \left ( \frac{1}{2} R_{AB}
-  \frac{1}{2}  {}^\gamma R_{AB} + s_A s_B - D_{(A} s_{B)} \right ) f =0\label{Eq}
\ee
where $R_{AB}:=R_{\mu\nu}|_{S_0} e^\mu_A e^\nu_B$ is the pull-back of the Ricci tensor to $S_0$, and
${}^\gamma R_{AB}$ is the Ricci tensor of the cut $(S_0,\gamma)$ itself.
\end{proposition}

This linear homogeneous PDE  is of second order and written in normal
form, i.e. with all second derivatives expressed in terms of lower order terms.
Its  integrability conditions were briefly considered in \cite{mps} and will be further analyzed in a forthcoming paper.
It follows immediately from its  structure that one can freely prescribe at most  $f |_p$ and $df|_p$ at $p \in S_{0}$.
Since the dimension of $S_{0}$ is $n-1$, the maximal number of linearly independent
solutions is $n$ which indeed is the maximal number of degenerate Killing generators. When the original $\xi$ is non-degenerate, they all
add up to the maximum possible dimension of $\Kill_\H$.  When $\xi$ is degenerate, consistency of the construction demands that $f= \mbox{const}$ is a solution
of (\ref{Eq}) so that, in particular, adding $\xi$ to the set of Killing generators
arising  from (\ref{Eq}) does not increase the dimension. This can be checked explicitly because, from the Gauss equation for the cut $(S_0,\gamma)$ and general identities for Killing horizons, one can show \cite{mps} that
(\ref{Eq}) is fully equivalent to
\begin{equation}
D_AD_B f - s_{A}D_{B}f -s_B D_A f +\k K_{AB} =0.
\label{Eq1}
\end{equation}
Thus, if $\H$ is fully degenerate, then in particular $\k=0$ and $f=$const.\ is a solution of the master equation (\ref{Eq1}) leading to the original Killing $\xi$. Hence, there are at most $n$ independent Killings in $\Kill_\H$ in this degenerate case, as it should.

\section{Construction of spacetimes with MKHs via characteristic Cauchy problems}
\label{sec4}

Let us now study the construction of $\Lambda$-vacuum spacetimes with MKHs in terms of characteristic initial value problems.
For this we first recall some useful results.  The \emph{null second fundamental form} of a null hypersurface $N$, relative to a field
$K$ of null tangents to $N$, is defined as
$$
{}^N\!K(X,Y) := g(X, \nabla_Y K)
\,, \quad \text{where} \quad X,Y\in T N \, .
$$
${}^N\!K$ is a symmetric tensor field on $N$ and shares the same degeneracy as the first fundamental form $\gamma$ of $N$: $\gamma(K,\cdot)={}^N\!K(K,\cdot)=0$.
The trace-free part of ${}^N\!K$ gives the \emph{shear} ${}^N\!\pi:=({}^N\!K)_{\mathrm{tf}}$ of $N$ relative to $K$, and its trace ${}^N\!\theta:= \mathrm{tr}({}^N\!K)$ is called  \emph{expansion}.

The spacelike cross sections of $N$ will usually be referred to as {\em cuts} in this paper. Given any cut $S\subset N$, we consider its torsion one-form $\bm{s}$
relative to a null normal frame $\{k,\ell\}$:
let $k$ and $\ell$ be two null normals to the codimension-2 spacelike cut $S$ normalized by $g(k, \ell)=-1$.
Then the \emph{torsion one-form} $\bm{s}$ of $S$ with respect to
$\{k,\ell\}$ is given by the formula
\begin{equation}
{\bm  s}(Z):=g(\nabla_Zk, {\ell})\,, \quad \text{where} \quad Z\in TS
\,.
\end{equation}
For later use we recall that under a boost of the null basis
$\{ k' = \beta k, \ell' = \beta^{-1} \ell\}$, $\beta$ a smooth function on $S$, the torsion one form transforms
as
\begin{align}
  \bm{s'} = \bm{s} + \beta^{-1} d\beta.
  \label{transformation}
\end{align}

\subsection{Characteristic initial value problem}
\label{sec_initial}

We recall a fundamental result by Rendall \cite{rendall}, presented here in a slightly more geometrical version, along the lines of the discussion in section 4 in \cite{manyways}.
\begin{theorem}[\cite{rendall}]
\label{thm_Cauchy_problem}
Let $M_a$ be two smooth $n$-dimensional manifolds intersecting on a smooth $(n-1)$-dimensional manifold $S \equiv M_1\cap M_2$. Let $g^\circ$ be a degenerate semi-positive symmetric 
tensor on $M_1\cup M_2$ with one-dimensional radical and continuous and non-degenerate at $S$. (Equivalently, with signature $(0,+,\dots ,+)$). Let also $\bm{\torsion}$ be a given smooth one-form on $S$, and $\nu>0$, $\omega >0$, $\omega_u$ and $\omega_v$ be given smooth functions on $S$. Then, there exists an
$(n+1)$-dimensional Lorentzian manifold with boundary $(U,g)$ whose smooth metric $g$ is a solution of the $\Lambda$-vacuum Einstein field equations
$$
R_{\mu\nu} -\frac{1}{2} R g_{\mu\nu} +\Lambda g_{\mu\nu} =0
$$
 and a unique non-vanishing function $\Omega$ on $U$, such that
\begin{itemize}
\item the boundary is formed by two intersecting null hypersurfaces $N_1$ and $N_2$ to the past of $(U,g)$ which are isometric to some open neighborhood of $(M_1,\Omega^2 g^\circ)$ and $(M_2,\Omega^2 g^\circ)$ around $S$, respectively. We keep the notation $S$ for the intersection of $N_1$ and $N_2$ in $U$.
\item there exist representatives $\K$ and $\L$ of the
null generators of $N_1$ and $N_2$, respectively, with $g(\K,\L)|_S =\nu $
\item $\bm{\torsion}$ is the torsion one-form relative to $k=\K/\sqrt{\nu}$ and
  $\ell = -\L/\sqrt{\nu}$ on $S$
\item On $S$, one has $\Omega|_S =\omega$, $\K(\Omega) |_S = \omega_u$ and $\L (\Omega)|_S = \omega_v.$
\end{itemize}
Moreover, any two such Lorentzian manifolds are isometric on some neighbourhood of $S$.
\end{theorem}

\begin{remark}
{\rm
  In fact, the solution can be extended so that its boundary includes the
 whole initial hypersurfaces 
 provided the Raychaudhuri equation for the initial-data expansions
 admits a global solution \cite{luk, CCTW}.
}
\end{remark}

We introduce (cf.\ \cite{rendall}) an \emph{adapted null coordinate system} $(x^{\mu})=(u,v, x^A)$, $A=3,\dots,n+1$, where
$N_1=\{v=0\}$ and $N_2=\{u=0\}$. In particular, $S=\{u=0, v=0\}$. Moreover, we assume that $u$ (resp.\ $v$) parameterizes the null geodesic
generators of  $N_1$ (resp.\ $N_2$) , while the $x^A$'s are local coordinates on the $\{v=0, u=\mathrm{const.}\}$- and
$\{u=0, v=\mathrm{const.}\}$-level sets, that is, the respective adapted cuts.

The induced metric on each cut
$\one{S}_u := \{ u = \mbox{const} \} \subset N_1$ and
$\two{S}_v:= \{ v = \mbox{const} \} \subset N_2$ is denoted by
$\widehat{g}$ and similarly the associated covariant derivative,
connection coefficients, Ricci tensor etc.\ carry a hat.
The induced Riemannian metric on the intersection surface $S$ will still be denoted by $\gamma$ 
and, correspondingly, connection coefficients, Ricci tensor etc.\ will be decorated with $\gamma$. Its covariant derivative will be denoted by $D$ as before.

In adapted null coordinates we take $\K:= \partial_u$ on $N_1$ and
$ \L:= \partial_v$ on $N_2$. Then the non-vanishing components of the respective null second fundamental forms are given by
\begin{equation*}
{}^{(1)}K_{AB} = \frac{1}{2}\partial_u g_{AB}
\,,
\quad
{}^{(2)}K_{AB} =\frac{1}{2}\partial_v g_{AB}\,.
\end{equation*}
Note that in these adapted coordinates, the metric component $g_{uv}$ on $S$ reads $g_{uv}|_{S} = \nu$ and the torsion one-form $\bm{\torsion}$ relative to $k$ and
$\ell$ of Theorem \ref{thm_Cauchy_problem} is \cite{manyways}
%
\begin{align}
  \torsion_A\overset{S}{=} \frac{1}{2}(\Gamma^v_{vA}-\Gamma^u_{uA})\,.
  \label{exp:varsigma}
\end{align}

%
%

\begin{remark}
{\rm
For our purposes it is more convenient to  prescribe the shears ${}^{(a)}\pi_{AB}$ of $N_a$, $a=1,2$, together
with $\gamma_{AB} := g_{AB}|_S$ and the expansions of $N_1$ and $N_2$ at $S$, ${}^{(a)}\theta|_S$, rather than
the family $ g^\circ_{AB}$ together with $\Omega$, $\partial_u\Omega$ and $\partial_v\Omega$ at $S$.

Note for this that the shear is determined from $ g^\circ_{AB}$ by a first-order ODE. The remaining freedom to prescribe
$ g^\circ_{AB}|_S$ together with that to choose $\Omega|_S$ can be combined into the prescription of  $\gamma_{AB}$,
i.e.\ the induced Riemannian metric on $S$.
The freedom to prescribe   $\partial_u\Omega$ and $\partial_v\Omega$ at $S$ can be identified with that to choose
the expansions ${}^{(a)}\theta$ of $N_a$ at $S$.
}
\end{remark}

Some remarks concerning the gauge freedom are in order. Usually when solving the evolution problem
one imposes a (generalized) wave-map gauge condition \cite{f1, CCM2}.
In this paper, though, it is irrelevant how the coordinates are extended off the initial hypersurface.
We will nevertheless exploit some of the gauge freedom which arises from the freedom to prescribe the adapted null coordinates
on the initial hypersurface: There remains the gauge freedom to reparametrize the null geodesic generators
of $N_1$ and $N_2$, which can be used to prescribe the connection coefficients $\Gamma^u_{uu}|_{N_1}$ and $\Gamma^v_{vv}|_{N_2}$ \cite{CCM2}. They will vanish if and only if $u$ and $v$ are affine parameters on $N_1$ and
$N_2$, respectively. The remaining reparametrization freedom can be used to prescribe $g_{uv}|_S > 0$ and add
gradients to the torsion one-form $\varsigma$ \cite{KIDs}.

The $\Lambda$-vacuum equations imply transport equations of certain fields along
the null geodesic generators of $N_a$. Here we provide the relevant equations which will be used later on.
It is convenient to set
$$
{}^{(1)}\Xi_{AB}\overset{N_1}{:=}2\Gamma^{u}_{AB}+{}^{(1)}K_{AB}g^{uu}
\,, \quad
{}^{(2)}\Xi_{AB}\overset{N_2}{:=}2\Gamma^{v}_{AB}+{}^{(2)}K_{AB}g^{vv}.
$$
%
%
In geometrical terms, $\frac{1}{2} \one{\Xi}_{AB}$ is the null second fundamental
form of the sections $\{ u=\mbox{const}\}$ along the null normal $\underline{\K}$
satisfying $\la \K,\underline{\K} \ra = -1$, and similarly
for $\two{\Xi}_{AB}$.

The following equations hold, which we present for reasons of definiteness on $N_1$, cf.\ \cite{CCM2, manyways, rendall},
\begin{eqnarray}
\Big(\partial_u +\frac{{}^{(1)}\theta}{n-1}- \Gamma^u_{uu}\Big){}^{(1)}\theta+|{}^{(1)}\pi|^2  &\overset{N_1}{=}& 0
\,,
\label{constr1}
\\
(\partial_u + {}^{(1)}\theta )\Gamma^u_{uA} - \frac{n-2}{n-1} \partial_A{}^{(1)}\theta -\partial_A\Gamma^u_{uu}+\widehat\nabla_B{}^{(1)}\pi_A{}^B  &\overset{N_1}{=}&0
\,,
\label{constr2}
\\
(\partial_u+ {}^{(1)}\theta + \Gamma^u_{uu})\mathrm{tr}({}^{(1)}\Xi)- 2g^{AB}(\widehat \nabla_A+ \Gamma^u_{uA})\Gamma^u_{uB}+\widehat R &\overset{N_1}{=}& 2\Lambda
\,.
\label{constr3}
\end{eqnarray}
As mentioned above, $\widehat \nabla$ and $\widehat R $ refer to the Levi-Civita covariant derivative and Ricci scalar of the Riemannian family
 $\widehat g=  g_{AB}\mathrm{d}x^A\mathrm{d}x^B|_{N_1\cup N_2}$, and $|{}^{(1)}\pi|^2 :={}^{(1)}\pi_A{}^B{}^{(1)}\pi_B{}^A$.
 Moreover, from $(R_{AB})_{\mathrm{tf}}=0$ we find
%
\begin{eqnarray}
\Big(\partial_{u}+\frac{n-5}{2(n-1)}{}^{(1)}\theta + \Gamma^u_{uu}\Big)({}^{(1)}\Xi_{AB})_{\mathrm{tr}}
-2({}^{(1)}\pi_{(A}{}^C({}^{(1)}\Xi_{B)C})_{\mathrm{tf}})
&&
\nonumber
\\
- \frac{1}{2}
  \mathrm{tr}({}^{(1)}\Xi) {}^{(1)}\pi_{AB}
-2(\widehat\nabla_{(A}\Gamma^u_{B)u})_{\mathrm{tf}}
 -2(\Gamma^{u}_{Au}\Gamma^{u}_{Bu})_{\mathrm{tf}}+( \widehat R_{AB})_{\mathrm{tf}}&\overset{N_1}{=}&0
\,.\phantom{xx}
\label{constr4}
\end{eqnarray}
The initial data for \eq{constr1} are part of the free data.
In contrast, the initial data for \eq{constr2} are determined by the torsion one-form and $g_{uv}|_S$, while those for \eq{constr3} and \eq{constr4}
are determined by expansion and shear at $S$, respectively, of the other null hypersurface,
\begin{equation}
{}^{(1)}\Xi_{AB} \overset{S}{=}-\frac{2}{\nu}  {}^{(2)}K_{AB} \,, \quad
{}^{(2)}\Xi_{AB}\overset{S}{=} - \frac{2}{\nu} {}^{(1)}K_{AB}
\,,
\end{equation}
which follows immediately from the fact that, on $S$,
$\underline{\K} = - \frac{1}{\nu} \L$ and
$\underline{\L} = - \frac{1}{\nu} \K$.

We also need the following fact concerning torsion one-forms on the cuts
$\one{S}_u$ and $\two{S}_v$. Denoting by $\one{s}_A$  (resp.
$\two{s}_A$)
the torsion
one-form with respect to $\{ \K,\underline{\K} \}$
(resp. $\{ \L,\underline{\L} \}$) one easily checks that
  \begin{align}
    \one{s}_A \stackrel{\one{S}_u}{=} - \Gamma^u_{uA}, \quad
    \quad
    \two{s}_A \stackrel{\two{S}_v}{=} - \Gamma^v_{vA}.
    \label{ss}
  \end{align}
    In particular, this means
  that these Christoffel symbols are in fact tensorial on each cut  so
  $\widehat{\nabla}$-covariant derivatives (as in e.g. (\ref{constr3}) or
  (\ref{constr4})) make sense.  The following identities on $S$
follow directly from the definitions  and from (\ref{exp:varsigma}), respectively
\begin{align}
  - \frac{1}{\nu} \partial_A \nu \stackrel{S}{=} \one{s}_A + \two{s}_A,
  \quad \quad \quad
  \torsion_A \stackrel{S}{=} \frac{1}{2} \left ( \one{s}_A - \two{s}_A \right ).
  \label{torsion_forms_S}
 \end{align}

\subsection{KID equations}
\label{sec_KIDs}

We are interested in generating $\Lambda$-vacuum spacetimes which admit Killing vectors $\zeta$ in terms
of an  initial value problem.  The spacetime Killing equation $\mcL_{\zeta} g=0$ is then replaced by the
\emph{Killing initial data (KID) equation}.  If and only if it admits $m$ independent solutions,
the emerging  vacuum spacetime admits $m$ independent Killing vectors.
For a characteristic initial value problem they have been derived in \cite{KIDs}:

\begin{theorem}[\cite{KIDs}]
Consider two smooth hypersurfaces $N_a$, $a=1,2$, in an $(n+1)$-dimensional manifold, $n\geq 3$, with transverse intersection along a smooth $(n-1)$-dimensional submanifold $S$ in adapted null coordinates.
Let $ \zeta$
 be a continuous vector field defined on
$N_1\cup N_2$ such that $\zeta|_{N_1}$ and $\zeta|_{N_2}$ are smooth. Then $\zeta$ extends to a smooth vector field  on $D^+(N_1\cup N_2)$ satisfying
the Killing equation and coinciding with $ \zeta$ on $N_1\cup N_2$  if and only if the KID equations hold by which we mean that on $N_1$
(note that the equations on  $N_1$ and $N_2$ do not involve transverse derivatives of $ \zeta$)
\begin{eqnarray}
\nabla_u\zeta_u  &\overset{N_1}{=}& 0
\,,
\label{KID_gen1}
\\
\nabla_{(u} \zeta_{A)}  &\overset{N_1}{=}& 0
\,,
\label{KID_gen2}
\\
(\nabla_{(A}\zeta_{B)} )_{\mathrm{tf}} &\overset{N_1}{=}& 0
\,,
\label{KID_gen3}
\\
\nabla_u\nabla_u \zeta^u &\overset{N_1}{=}& R_{\mu uu}{}^u \zeta^{\mu}
\,,
\label{KID_gen4}
\end{eqnarray}
with identical corresponding conditions on $N_2$, while on $S$ one further needs to assume 
\begin{eqnarray}
\nabla_{(u}\zeta_{v)} &\overset{S}{=}& 0
\,,
\label{KID_gen5}
\\
g^{AB}\nabla_A\zeta_B &\overset{S}{=}& 0
\,,
\label{KID_gen6}
\\
\partial_u(g^{AB}\nabla_A\zeta_B) &\overset{S}{=}& 0
\,,
\label{KID_gen7}
\\
\partial_v(g^{AB}\nabla_A\zeta_B) &\overset{S}{=}& 0
\,,
\label{KID_gen8}
\\
  \nabla_A\nabla_{[u}\zeta_{v]}
                      &\overset{S}{=}& R_{vuA}{}^{\mu} \zeta_{\mu}
\,.
\label{KID_gen9}
\end{eqnarray}
\end{theorem}

\subsection{Bifurcate Killing horizon}
\label{sec_bifurcate}

We consider a vacuum spacetime $( M ,g)$ in $n+1$ dimensions, $n\geq 3$, possibly with cosmological constant $\Lambda$, which admits a Killing vector $\xi$ which generates a bifurcate Killing horizon $\mathcal{H}=\mathcal{H}_1^+\cup\mathcal{H}_2^+\cup\mathcal{H}_1^-\cup\mathcal{H}_2^-\cup S$  according to Definition \ref{defi} and the comment that follows it.
The associated surface gravity is constant, $\kappa_{\xi}=\mathrm{const.}\ne 0$ on $\mathcal{H}$.
In fact, it is shown in \cite{re} that, adding a certain technical condition, any spacetime  with a non-degenerate Killing horizon with constant surface gravity 
can be extended to a spacetime where this Killing horizon forms a portion of a bifurcate Killing horizon.

For now, let us restrict attention to those parts of $\mathcal{H}$ which lie in the causal future of $S$,  that is to say, $\mathcal{H}^+_1$ and $\mathcal{H}^+_2$.
Then the spacetime $D^+(\mathcal{H}^+_1\cup\mathcal{H}^+_2)$ can be generated in terms of a characteristic initial value problem
for appropriately  data specified on $\mathcal{H}^+_1\cup\mathcal{H}^+_2\cup S$.

\begin{theorem}[\cite{ref, hiw, KIDs}]
\label{theorem_fund_form}
\begin{enumerate}
\item[(i)]
$D^+(\mathcal{H}^+_1\cup\mathcal{H}^+_2)$ admits a Killing vector field $\xi$ for which both, $\mathcal{H}^+_1$ and $\mathcal{H}^+_2$, are Killing horizons
if and only if the null second fundamental form ${}^{(a)}K$ vanishes on $\mathcal{H}^+_1$ and $\mathcal{H}^+_2$. In that case $\xi$ is unique (up to constant rescaling) and  the associated surface gravity  is constant and satisfies
$$
\kappa_{\xi}(\mathcal{H}^+_1)=-\kappa_{\xi}(\mathcal{H}^+_2)=\partial_{[u}\xi_{v]}|_S \ne 0\,,
$$
i.e.\ the bifurcate horizon in non-degenerate.
Moreover, $\xi \overset{S}{=}0$.
\item[(ii)] Assume that such a $\xi$ exists, then $D^+(\mathcal{H}^+_1\cup\mathcal{H}^+_2)$ admits a second independent Killing vector  $\zeta$,
tangent to $S$,  if and only
if the bifurcation surface $S$ with the Riemannian metric $\gamma\equiv g_{AB}\mathrm{d}x^A\otimes\mathrm{d}x^B|_S$ admits a non-trivial (i.e. not the zero vector)
Killing vector   $\ol  \zeta=\ol  \zeta^A\partial_A$ such that the $\ol  \zeta$-Lie-derivative of the torsion one-form of $S$ is exact.
The Killing vector $\zeta$ is then uniquely determined (up to an additive $c \xi$, $c\in\mathbb{R}$) by the condition $\zeta|_S=\ol  \zeta$,
 and is tangent to $\mathcal{H}^+_1\cup\mathcal{H}^+_2$.
\end{enumerate}
\end{theorem}
It follows immediately that there cannot exist a second independent Killing vector $\zeta$ with Killing horizons $\mathcal{H}^+_1$ \emph{and} $\mathcal{H}^+_2$
 as this would imply that $\ol \zeta := \zeta |_S =0$,
which contradicts the fact that $\ol \zeta$ has to be non-trivial (item (ii) of the theorem).
However, $\zeta$ may have either  $\mathcal{H}^+_1$ or  $\mathcal{H}^+_2$ as Killing horizon.
So let us analyze under which conditions one of the horizons,  $\mathcal{H}^+_1$ say, is simultaneously the Killing horizon of a second Killing vector field, which we denote by $\eta$. Proportionality of $\eta$ and
$\xi$ on $\H_1^+$ is equivalent,  
in adapted null coordinates, to
\begin{align}
\eta^v |_{\mathcal{H}_1^+}= 0 \,, \quad
\eta^u|_{\mathcal{H}_1^+}\ne0\,, \quad \eta^A|_{\mathcal{H}_1^+}=0
  \, .
  \label{proportionality}
\end{align}
%
At $\H_1^+$ the metric in adapted null coordinates is of the form
\begin{align}
  ds^2 \stackrel{\H_{1}^+}{=} dv \left ( g_{vv} dv + 2 g_{uv} du
  + 2 g_{vA} dx^A \right )
  + \widehat{g}_{AB} dx^A dx^B, \quad
  g_{uv} \stackrel{S}{=} \nu,  \enspace g_{vv} \stackrel{S}{=} g_{vA} \stackrel{S}{=} 0,
  \label{metric}
\end{align}
so an equivalent form of (\ref{proportionality}) is
\begin{equation}
\eta_v|_{\mathcal{H}_1^+}\ne  0 \,, \quad
\eta_u|_{\mathcal{H}_1^+}=0\,, \quad \eta_A |_{\mathcal{H}_1^+}=0
\,.
\label{horizon_conds}
\end{equation}
We shall analyze to what extent these additional conditions are  compatible with the KID equations.
For this we shall first rewrite and simplify the system of KID equations in the special case of  a bifurcate Killing horizon.


The null hypersurfaces $N_1$ and $N_2$ are required to form a bifurcate Killing horizon, i.e.\
 the null second fundamental forms need to vanish (cf.\ Theorem~\ref{theorem_fund_form}).
With regard to the characteristic initial value problem this is achieved by  the initial data
\begin{equation}
{}^{(a)}\pi_{AB} = 0\,, \quad  {}^{(a)}\theta \overset{S}{=}0\,, \quad a=1,2\,.
\end{equation}
In this setup we denote $N_1$ and $N_2$ by
$\mathcal{H}^+_1$ and $\mathcal{H}^+_2$.
The Raychaudhuri equation then implies ${}^{(a)}\theta =0$. 
Equation (\ref{constr2})
together with $\one{K}_{AB}=0$ and (\ref{ss}) becomes
\begin{align}
  \partial_u \one{s}_A + \partial_{A} \Gamma^{u}_{uu} =0
  \label{dersone}
\end{align}
From this equation it follows that the Riemann tensor components
$R_{Auu}{}^{\mu}$ vanish. Indeed, let $e_A := \partial_{x^A}$ so that
$[\K, e_A] =0$. The vanishing of $\one{K}_{AB}$ and the definition
of the torsion 1-form gives (cf. (\ref{torsion}))
\begin{align}
  \nabla_{e_A} \K \stackrel{\H_1^+}{=}
  - \one{s}_A \K.
  \label{nablaeK}
\end{align}
Inserting in the Ricci identity and using $\nabla_{\K} \K
\stackrel{\H_1^+}{=}
\Gamma^{u}_{uu} \K$ yields, on $\H_{1}^+$,
\begin{align}
  R_{Auu}{}^{\mu} =  R_{\alpha\beta\delta}{}^{\mu} e_A^{\alpha}
  \K^{\beta} \K^{\delta} =
  \left ( \nabla_{\K} \nabla_{e_A} - \nabla_{e_A} \nabla_{\K} -
  \nabla_{[\K,e_A]}
  \right ) \K^{\mu} = -
  \left (   \K(\one{s}_A)  + \partial_A \Gamma^u_{uu} \right ) \K^{\mu} = 0.
  \label{curv1}
\end{align}
This means in particular that the right hand side of (\ref{KID_gen4}) is
$R_{\mu uu}{}^u \zeta^{\mu} = R_{vuu}{}^u \zeta^v
= g_{uv}^{-1} R_{vuu}{}^u \zeta_u $, the second equality following from
\begin{align*}
  g^{v\alpha} \stackrel{\H_1^+}{=} g_{uv}^{-1} \delta^{\alpha}{}_{u}
\end{align*}
which is found by inverting the metric (\ref{metric}).
Using the formulas in \cite[Appendix~A]{CCM2} one checks that the
KID equations (\ref{KID_gen1})-(\ref{KID_gen4}) reduce to the following system
on $\mathcal{H}^+_1$ (and correspondingly on $\mathcal{H}^+_2$):
\begin{eqnarray}
(\partial_u - \Gamma^u_{uu}) \zeta_{u}  &\overset{\mathcal{H}^+_1}{=} & 0
\,,
\label{ODE_KID1}
\\
\partial_{u} \zeta_{A} + (\widehat\nabla_{A} -2\Gamma^u_{uA})\zeta_{u}  &\overset{\mathcal{H}^+_1}{=} &0
\,,
\label{ODE_KID2}
\\
(\widehat\nabla_{(A}\zeta_{B)} )_{\mathrm{tf}} -(\Gamma^u_{AB})_{\mathrm{tf}} \zeta_u &\overset{\mathcal{H}^+_1}{=} & 0
\,,
\label{constr_KID}
\\
\partial_u(\partial_u+ \Gamma^u_{uu}) \zeta^u
-\frac{1}{2}\zeta_u(\partial_u+ \Gamma^u_{uu})(\partial_u+2 \Gamma^u_{uu})g^{uu}
\nonumber
&&
\\
-g^{AB}\Gamma^u_{uA} \Gamma^u_{uB}\zeta_u
+g^{AB}\zeta_{B}\partial_u \Gamma^u_{uA}
+2g^{AB}\Gamma^u_{uA}\partial_u\zeta_{B}  &\overset{\mathcal{H}^+_1}{=} & g_{uv}^{-1} R_{vuu}{}^u \zeta_u
\,.
\label{ODE_KID3}
\end{eqnarray}
%
%
%
%

We next rewrite the KID equation on $S$. In particular we need to evaluate
the curvature
components in (\ref{KID_gen9}):
\begin{align*}
  R_{vuA\mu} \zeta^{\mu} \stackrel{S}{=}
  R_{vuAu} \zeta^u + R_{vuAv} \zeta^v +
  R_{vuAB} \zeta^B = R_{vuAB} \zeta^B,
\end{align*}
the first two terms being zero because of (\ref{curv1}) and the corresponding identity on $\H_2^+$. To evaluate the remaining term we first note that, on $\H_{1}^{+}$,
  \begin{align*}
    \la \underline{\K}, \nabla_{e_A} \nabla_{e_B} \K \ra
    = \nabla_{e_A} \la \underline{\K},
    \nabla_{e_B} \K \ra - \la \nabla_{e_A} \underline{\K},
    \nabla_{e_B} \K \ra =
    \nabla_{e_A} \one{s}_B - \one{s}_A \one{s}_B
\,,
  \end{align*}
  where we used (\ref{nablaeK}) and the corresponding
  $\nabla_{e_A}  \underline{\K} = \one{s}_A \underline{K}$.
  From $\partial_v \stackrel{S}{=} - \nu \underline{\K}$ and
  $[e_A,e_B] =0$ we have
  \begin{align}
    R_{vuAB} \stackrel{S}{=} - \nu \la \underline{ \K}, \nabla_{e_A} \nabla_{e_B}
    \K - \nabla_{e_B} \nabla_{e_A} \K \ra =
    \nu \left ( \partial_B \one{s}_A - \partial_A \one{s}_B \right )
    = \nu \left ( \partial_B \torsion_A - \partial_A \torsion_B \right ),
    \label{Riemann_component}
    \end{align}
  where in the last equality we used that boosting a null basis changes
  the torsion one-form with an exact one-form  (cf. (\ref{transformation})).

Set $\ol  \zeta:= \zeta^A\partial_A|_S$. On $S$ we then have
(cf.\ the corresponding computation in \cite{KIDs} but be aware that the  Killing vector was assumed there to be tangential
to $S$ whence some additional terms appear here) 
\begin{eqnarray}
  \left (  \partial_u +\Gamma^u_{uu} + \Gamma^{v}_{uv} \right )   \zeta^u
+  \left (  \partial_v +\Gamma^v_{vv} + \Gamma^{u}_{uv} \right )   \zeta^v
  +  \nu^{-1} \mcL_{\ol\zeta} \nu  &\overset{S}{=}&  0\, ,
  \label{KID_S1}
\\
D^A\zeta_A&\overset{S}{=}& 0
\,,
\\
D^A\partial_u\zeta_A-\Big(\gamma^{AB}\Gamma^u_{uA}\Gamma^u_{uB}
+D^A \Gamma^u_{uA} -\frac{{}^{\gamma}R}{2} + \Lambda\Big) \zeta_{u}&\overset{S}{=}& 0
\,,
\\
D^A\partial_v\zeta_A-\Big(\gamma^{AB}\Gamma^v_{vA}\Gamma^v_{vB}
+D^A \Gamma^v_{vA} -\frac{{}^{\gamma} R}{2} +  \Lambda\Big) \zeta_{v}&\overset{S}{=}& 0
\,,
\\
D_A(\partial_{v}+\Gamma^v_{vv}- \Gamma^u_{uv})\zeta^v-D_A(\partial_{u}+\Gamma^u_{uu}
-\Gamma^v_{uv})\zeta^u
+2\mcL_{\ol  \zeta}\varsigma_A
&\overset{S}{=}& 0
\,,
\label{KID_S5}
\end{eqnarray}
%
where indices on $D $ have been raised with $\gamma_{AB}$,
and where \eq{constr3} and (\ref{Riemann_component}) have been used.
%
For convenience let us  impose the following gauge conditions: require $u$ and $v$ to be affine parameters along the
null geodesic generators of $\mathcal{H}^+_1$ and $\mathcal{H}^+_2$, respectively. Then  (cf.\ \cite{CCM2})
\begin{equation}
 \Gamma^u_{uu}  \overset{\mathcal{H}^+_1}{=} 0\,, \quad  \quad \Gamma^v_{vv} \overset{\mathcal{H}^+_2}{=} 0
\,.
\label{cond_gauge}
\end{equation}
This leaves the freedom to do affine transformations of these coordinates. With the gauge choice $S=\{u=0, v=0\}$
there remains the freedom to do rescalings of the form $u\mapsto f^{(u)}(x^A) u$ and  $v\mapsto f^{(v)}(x^A) v$ which can be partially fixed  \cite{KIDs} to achieve that
\begin{equation}
\nu \overset{S}{=} 1
\,.
\label{cond_gauge2}
\end{equation}
The remaining coordinate freedom is $\{ \widetilde{u} = \phi^{-1}(x^A) u\} $ and
$\{ \widetilde{v} = \phi(x^A) v\}$. Its effect on $\torsion$ is
\begin{align}
  \widetilde{\bm{\torsion}} = \bm{\torsion} + \phi^{-1} d\phi
\,,
  \label{transtorsion}
\end{align}
because of (\ref{transformation}) together with
\begin{align*}
  \partial_{\widetilde{u}} \stackrel{S}{=} \phi \partial_u, \quad \quad
\partial_{\widetilde{v}} \stackrel{S}{=} \phi^{-1} \partial_u \, .
\end{align*}
Thus, all the remaining gauge freedom is to add gradients to $\bm{\torsion}$.

\begin{remark}\label{rem}
{\rm
In fact, the non-gauge part of the torsion one-form $\torsion$ together with the Riemannian metric $\gamma$ on $S$  provide the only ``physical'', non-gauge data for a characteristic initial value problem for the vacuum equations on a bifurcate Killing horizon.
}
\end{remark}

We want to compute the $u$-derivative of \eq{constr_KID}.
%
With \eq{constr4}, \eq{ODE_KID1} and \eq{ODE_KID2} we   obtain on $\mathcal{H}^+_1$
\begin{eqnarray*}
0 &\overset{\mathcal{H}^+_1}{=}&
(\widehat\nabla_{(A}\partial_u\zeta_{B)} )_{\mathrm{tf}} -\partial_u(\Gamma^u_{AB})_{\mathrm{tf}} \zeta_u
 -(\Gamma^u_{AB})_{\mathrm{tf}} \partial_u \zeta_u
\\
&\overset{\mathcal{H}^+_1}{=}&
-(\widehat\nabla_{A} \widehat\nabla_{B}\zeta_{u} )_{\mathrm{tf}}
+(\widehat\nabla_{(A}\Gamma^u_{B)u} )_{\mathrm{tf}} \zeta_{u}
+2 ( \Gamma^u_{u(A}\widehat\nabla_{B)}\zeta_{u} )_{\mathrm{tf}}
-(\Gamma^{u}_{Au}\Gamma^{u}_{Bu})_{\mathrm{tf}}\zeta_u
+\frac{1}{2}(\widehat R_{AB} )_{\mathrm{tf}}\zeta_u
\,.
\end{eqnarray*}
Observe that
$\widehat g$, $\widehat R_{AB}$, $\Gamma^u_{uA}|_{\mathcal{H}^+_1}$ and $\zeta_u|_{\mathcal{H}^+_1}$ are $u$-independent
(cf.\  \eq{constr2} and \eq{ODE_KID1}),
so this equation is independent of $u$ and therefore only needs to be satisfied at $S$ (similarly on $\mathcal{H}^+_2$).
 In fact, this is a crucial property which allows to reduce  the KID equations to  a system of equations on $S$ in the special case
where $S$ is a bifurcation surface (cf.\ the different setting in Section~\ref{sec_fully_deg}).

Using 
\eq{ODE_KID2}, \eq{constr_KID} and \eq{ss} with the fact that in the present gauge $\one{s}_A = - \two{s}_A = \varsigma_A$ on $S$ (cf. \eq{torsion_forms_S}),
one checks that the system
\begin{eqnarray}
D_{(A}\zeta_{B)} &\overset{S}{=}& 0
\,,
\\
\Big(\Delta_{\gamma} - 2 \torsion_AD^A
-D^A \torsion_A+ |\torsion|^2-\frac{{}^{\gamma} R}{2} +\Lambda\Big) \zeta_{v}
&\overset{S}{=}&0
\,,
\label{trace1}
\\
\Big(\Delta_{\gamma} + 2 \torsion_AD^A
+D^A \torsion_A+ |\torsion|^2-\frac{{}^{\gamma} R}{2} +\Lambda\Big) \zeta_{u}
&\overset{S}{=}&0
\,,
\label{trace2}
\\
\Big(D_{A} D_{B}
-2 \torsion_{(A}D_{B)}
-D_{(A}\torsion_{B)}
+\torsion_A\torsion_B
-\frac{1}{2}{}^ {\gamma} R_{AB} \Big)_{\mathrm{tf}}\zeta_v
&\overset{S}{=}&
0
\,,
\label{tracefree1}
\\
\Big(D_{A}D_{B}
+2 \torsion_{(A}D_{B)}
+D_{(A}\torsion_{B)}
+\torsion_A\torsion_B
-\frac{1}{2}{}^{\gamma} R_{AB} \Big)_{\mathrm{tf}}\zeta_u
&\overset{S}{=}&
0
\,,
\label{tracefree2}
\\
\partial_u\zeta^u+\partial_v\zeta^v+\Gamma^u_{uv}\zeta^v+\Gamma^v_{uv}\zeta^u
&\overset{S}{=}& 0
\,,
\label{tran_deriv2}
\\
D_A\Big(\partial_{v}\zeta^v-\partial_{u}\zeta^u
-\Gamma^u_{uv}\zeta^v+\Gamma^v_{uv}\zeta^u\Big)
+2\mcL_{\ol  \zeta}\torsion_A
&\overset{S}{=}&0
\label{tran_deriv}
\end{eqnarray}
is equivalent to  \eq{KID_S1}-\eq{KID_S5} plus the restriction of \eq{constr_KID} to $S$ as well as its $u$- and $v$-
derivatives.
Supposing that this  system admits a solution, the restriction of the Killing vector field $\zeta$ of the emerging vacuum spacetime to
$\mathcal{H}_1^+$ is computed  by solving the ODEs \eq{ODE_KID1}, \eq{ODE_KID2} and \eq{ODE_KID3}, and correspondingly on $\mathcal{H}_2^+$.
The initial data of the transport equations along the horizons is determined from the solution of
the system on $S$ as follows: assume a solution of the system
\eq{trace1}-\eq{tracefree2} has been selected. This fixes completely the initial data
for \eq{ODE_KID1}, \eq{ODE_KID2} and the corresponding equations of $\H_2^+$.
Now,  \eq{tran_deriv2}  can be read as an algebraic equation for $\partial_u\zeta^u + \partial_v\zeta^v|_S$ (equivalently $\partial_{(u}\zeta_{v)}|_S$) and
equation \eq{tran_deriv} is solvable if and only if $\mcL_{\ol  \zeta}{\bf s}$ is exact, and in that case it yields an equation algebraically solvable for
$\partial_u\zeta^u - \partial_v\zeta^v|_S$  (equivalently $\partial_{[u}\zeta_{v]}|_S$) in terms of an arbitrary additive constant. The additive constant
corresponds to the addition to $\zeta$ on $D^+(\H_1^+ \cup \H_2^+)$ of a constant
times $\xi$ (the Killing vector with respect to which $\H_1^{+} \cup \H_2^+$ is
a bifurcate horizon).  Indeed, for any
  constant $\alpha$ the vector field along
  $\H_1^+ \cup \H_2^+$ given by
  \begin{align*}
    \zeta_u \stackrel{\H_1^+}{=} \zeta_A \stackrel{\H_1^+}{=}    0,
    \quad
    \zeta^u  \stackrel{\H_1^+}{=} \alpha u, \quad \quad
    \quad \quad
    \zeta_v \stackrel{\H_2^+}{=} \zeta_A \stackrel{\H_2^+}{=}    0,
        \quad
    \zeta^v  \stackrel{\H_2^+}{=} -\alpha v,
  \end{align*}
  solves the full set of KID equations. Denoting by
  $\xi$ the Killing vector it defines on $D^+(\H_1^+ \cup \H_2^+)$, it is
  obvious that $\xi$ has $\H_1^+ \setminus S$ and $\H_2^+ \setminus S$ as
  Killing horizons and that it vanishes on $S$. It is also immediate from
  item (i) in theorem (\ref{theorem_fund_form})  that the surface gravities  are
  \begin{align*}
    \kappa_{\xi}(\mathcal{H}^+_1)=-\kappa_{\xi}(\mathcal{H}^+_2)= \alpha.
  \end{align*}

  Combining further \eq{trace1} \& \eq{tracefree1} and  \eq{trace2} \&  \eq{tracefree2}, respectively, into a single equation we have established the following
  theorem:

\begin{theorem}
\label{prop_KID_equations}
The data $(\gamma, \bm{\torsion})$ on $S$, supplemented by vanishing null second fundamental forms on $\mathcal{H}^+_1$ and $\mathcal{H}^+_2$ generate an $(n+1)$-dimensional $\Lambda$-vacuum spacetime with bifurcation surface $S$ with $k+1$ Killing vectors if and only if
the following set of KID equations admits $k$ independent solutions $(\zeta_u,\zeta_v, \zeta_A)|_S$:
\begin{align}
\text{$\ol \zeta\equiv \zeta^A\partial_A|_S$ is a Killing vector on $(S,\gamma)$ such that $\mcL_{\ol \zeta}{\bm{\torsion}}$ is exact}
\,,
\label{KID1}
\\
\Big(D_{A}D_{B}
-2  \torsion_{(A}D_{B)}
-D_{(A}\torsion_{B)}
+\torsion_A\torsion_B
-\frac{1}{2}{}^{\gamma} R_{AB}+\frac{\Lambda}{n-1} \gamma_{AB} \Big)\zeta_v
\overset{S}{=}
0
\,,
\label{KID2}
\\
\Big(D_{A}D_{B}
+2 \torsion_{(A}D_{B)}
+D_{(A}\torsion_{B)}
+\torsion_A\torsion_B
-\frac{1}{2}{}^{\gamma} R_{AB}+\frac{\Lambda}{n-1} \gamma_{AB}  \Big)\zeta_u
\overset{S}{=}
0
\,.
\label{KID3}
\end{align}
\end{theorem}

\begin{remark}
\label{rem_max_sym}
{\rm
\begin{enumerate}
\item[(i)]
  The KID equations \eq{KID2}-\eq{KID3} are identical in form
  to the master equation \eq{Eq} in $\Lambda$-vacuum.
    It should be emphasized however, that the two
  equations are {\it not} the same because the master equation holds
  for sections of  the Killing horizon and $S$ is not part of the horizons.
  In fact,  the torsion one-form $\bm{s}$ in \eq{Eq} is defined with  respect to
  $\xi$, so it makes no sense on $S$, where $\xi$ vanishes.
  The underlying reason why the two equations are the same will be discussed
  in the next subsection.
\item[(ii)]
As already mentioned, the trivial solution $\zeta|_S=0$ of this system, supplemented by  $\partial_{[u}\zeta_{v]}|_S=\mathrm{const.}\ne 0$ (cf.\ \eq{tran_deriv}) corresponds to data for the Killing vector  $\xi$ which generates the bifurcate horizon (a vanishing surface gravity would produce the trivial vector field).
\item[(iii)]
For the emerging spacetime to be maximally symmetric, \eq{KID1}-\eq{KID3} need to admit $\frac{1}{2}(n+1)(n+2)-1$ independent solutions.
It is shown in \cite{mps}
that each of \eq{KID2}-\eq{KID3} admits at most $n$ solutions, so a necessary condition is that the bifurcation surface $(S, \gamma)$ admits
$\frac{1}{2}n(n-1)$ Killing vectors, i.e.\   is maximally symmetric.
\end{enumerate}
}
\end{remark}

\subsubsection{Comparison with the master equation}

\label{sec_comparison}

There are two immediate differences between the initial data equation
(\ref{KID2}) and the master equation (\ref{Eq}). First, the torsion-one forms are a priori different, and, second, the former equation involves one of the
components of the Killing vector in adapted null coordinates and the latter involves the proportionality function between two Killings. Nevertheless, the two equations can be related to each other.

We discuss first the
issue of the different torsion one-forms. In the gauge where
$\nu =1$ we have shown before that
\begin{align*}
  \xi \stackrel{\H_1^+} = \kappa_{\xi} u \partial_u
\end{align*}
were $\kappa_{\xi}$ is the surface gravity of the bifurcate Killing horizon
with respect to $\xi$. On each cross section $\one{S}_u$ the proportionality
between $\xi$ and $\K$ is constant, and by (\ref{transformation})
we have $\bm{s} =
\bm{\one{s}}$.
 Moreover, $\bm{\one{s}}$ is independent of $u$ by
(\ref{dersone}) and since $\bm{\one{s}} = \bm{\torsion}$ on $S$ we conclude that actually $\bm{s} = \bm{\torsion}$. Had we chosen another cross section
of $\H_+^1$, this is defined implicitly by a positive graph function
$\psi : S \rightarrow \mathbb{R}$ as $S_{\psi}:= \{ u = \psi\}$. Defining $\tau :
  \H_1^+ \rightarrow \mathbb{R}$ as
  \begin{align*}
    \tau = \frac{1}{\kappa_{\xi}} \ln u
  \end{align*}
  we have $\xi(\tau) = 1$. Applying \cite[Lemma 3]{MarsTotallyGeodesic} the  torsion one-form
  of $S_{\psi}$ is
  \begin{align*}
    \bm{s}[\psi] = \bm{s} + \psi^{-1} d\psi
  \end{align*}
  and the freedom in changing the section corresponds to the freedom
  of adding differentials to $\bm{\torsion}$. Indeed,  the coordinate
  transformation $\{ \widetilde{u} = \psi^{-1} u, \widetilde{v} = \psi v\}$ preserves
  $\nu=1$, transforms the torsion
  $\bm{\torsion}$ as in (\ref{transtorsion}) and
  makes $S_{\psi} = \{ \widetilde{u} = 1 \}$, so that now we have a constant
  $\widetilde{u}$ section and by the previous argument $\bm{s}[\psi] = \bm{
    \widetilde{\torsion}}$ must hold.

  Let us compute how the KID equations on $S$ behave under the transformation
(\ref{transtorsion}).
Clearly, $\mcL_{\ol \zeta} \bm{\torsion}$ is exact if and only if $\mcL_{\ol \zeta}\widetilde{\bm{\torsion}}$ is exact, i.e.\ \eq{KID1} remains invariant as is should be.
One further finds that $\zeta_v$ is a solution of \eq{KID2} if and only if
$\psi \zeta_v$ is a solution of \eq{KID2} with $\bm{\torsion}$ replaced
by $\widetilde{\bm{\torsion}}$.
In particular, whenever \eq{KID2} admits a solution with no zeros we can find a gauge where this solution is given by $\zeta_v=1$,
and in this gauge the following relation holds
\begin{equation}
D_{(A}s_{B)}
-s_As_B+\frac{1}{2}{}^{\gamma} R_{AB}-\frac{\Lambda}{n-1} \gamma_{AB} =0
\,.
\label{special_gauge_master}
\end{equation}
%
If $\zeta_v$ has zeros, the Killing vector it generates will vanish at these points, which therefore do not belong to its
Killing horizon.  Restricting $\zeta_v$ to the domain where it does not have zeros, the above rescaling can be done.
In the gauge \eq{special_gauge_master} the KID equation \eq{KID2} becomes
\begin{equation}
\Big(D_{A}D_{B}
-2 s_{(A}D_{B)}
 \Big)\zeta_v
\overset{S}{=}
0
\,.
\label{Master_Equation}
\end{equation}
Again we find a deep connection with the master equation \cite{mps} as given in
\eq{Eq} because the null second fundamental form vanishes on
a bifurcation surface. We emphasize once more that, despite their identical
form, the equations are intrinsically different as they involve different objects and are constructed on different types of surfaces.

We now turn into the issue that the master equation in \cite{mps} involves
the proportionality factor $f$  between $\eta$ and $\xi$
while equation ({\ref{KID2}) is for a component of
  $\eta$ at the bifurcation surface. We use the proportionality
  (\ref{dfn_eqn_f}) in Lemma \ref{lemmamps} together with
  $\tau = \kappa_{\xi}^{-1} \ln u$ to conclude that
  \begin{align}
    \eta \stackrel{\H_{1}^+}{=} u^{-1} f_{\eta}  \xi
    = \kappa_{\xi} f_{\eta} \partial_u \label{etaxi}
  \end{align}
  and recall that $f_{\eta}$ is constant along $\xi$, hence independent of
  $u$. Expression (\ref{etaxi}) extends to the bifurcation surface and we
  get $\eta_v = \eta^u = \kappa_{\xi} f_{\eta}$ on $S$. On the other hand, evaluating on the cross
  section   $\{ u =u_0 \}$ with $u_0$ a positive constant,
  the proportionality function $f$ between
  $\eta$ and $\xi$ is $f = f_{\eta} u_0^{-1} = \eta_v |_{S} \kappa_{\xi}^{-1} u_0^{-1}$. The master equation
  (\ref{Eq}) is linear so the multiplicative constant
  $\kappa_{\xi}^{-1} u_0^{-1}$ can be dropped. This, together
  with the equality $\bm{\one{s}} = \bm{\torsion}$ explains the
  full coincidence  in form of the KID equation (\ref{KID2}) and the master
  equation (\ref{Eq}), while at the same times makes it clear the
  intrinsic  difference between the two (observe that $f = f_{\eta} u_0^{-1}$ makes no sense at $u_0=0$).

%

\section{Non-degenerate MKHs}
\label{sec_non-deg}

Let us return to our original question. For $\mathcal{H}^+_1$ to be  a multiple Killing horizon
there needs to exist at least one  Killing vector $\eta$   
which satisfies \eq{horizon_conds} in order to be tangential and null on $\mathcal{H}^+_1$ and independent of $\xi$.
%
%
It follows straightforwardly from \eq{ODE_KID1}-\eq{ODE_KID2}
that this will be the case if and only if
\begin{equation}
\eta_v |_{S}\ne 0
\,, \quad
\eta_u |_{S}= 0
\,, \quad \eta_A|_S= 0
\,.
\label{cond_double_horizon2}
\end{equation}
In that case the  KID equations \eq{KID1} and \eq{KID3} become trivial.
Whenever \eq{KID2}   admits  a non-trivial solution it can be extended
to a Killing vector field $\eta$ for which  $\mathcal{H}_1^+$ (but not $\mathcal{H}_2^+$) is a Killing horizon and, solving the corresponding problem into the past, also  $\mathcal{H}_1^-$.

As consequence of Theorem~\ref{prop_KID_equations} we find that the master equation is in this setting not only necessary but also
sufficient for the Killing horizon to be multiple.
\begin{corollary}
\label{cor_multiple_horizons}
Let  $( M ,g)$ be an $(n+1)$-dimensional $\Lambda$-vacuum spacetime, $n\geq 3$, which admits a Killing vector $\xi$ which generates  a   bifurcate Killing horizon $\mathcal{H}=\mathcal{H}_1^+\cup\mathcal{H}_2^+\cup\mathcal{H}_1^-\cup\mathcal{H}_2^-\cup S$ (which is then  non-degenerate).  Let
$(\gamma,\bm{\torsion})$ be the free data on $S$.
Then $\mathcal{H}_1^{\pm}$  is a MKH of order $m$
if and only if the following equation admits $m-1$ independent non-identically zero solutions $f$ on $S$,
\begin{equation}
\Big(D_{A} D_{B}
-2  \torsion_{(A}D_{B)}
-D_{(A} \torsion_{B)}
+\torsion_{A}\torsion_{B}
-\frac{1}{2}{}^{\gamma} R_{AB}+\frac{\Lambda}{n-1} \gamma_{AB} \Big) f
= 0
\,.
\label{double_horizon_cond}
\end{equation}
Correspondingly,  $\mathcal{H}_2^{\pm}$  is a multiple Killing horizon of order $m$
if and only if the following equation admits $m-1$ independent non-trivial solutions $ f $ on $S$,
\begin{equation}
\Big(D_{A} D_{B}
+2 \torsion_{(A}D_{B)}
+D_{(A}\torsion_{B)}
+\torsion_A\torsion_B
-\frac{1}{2}{}^{\gamma} R_{AB}+\frac{\Lambda}{n-1} \gamma_{AB}  \Big) f
= 0
\,.
\label{double_horizon_cond2}
\end{equation}
\end{corollary}

\begin{remark}
  {\rm
  As described in Section~\ref{sec_comparison}, whenever  \eq{double_horizon_cond}  (similarly for \eq{double_horizon_cond2}) admits a solution
with no zeros, one can globally introduce a gauge where \eq{special_gauge_master} holds. This equation appears in the context of near-horizon geometries and  is sometimes called NHG equation. Actually, \eq{double_horizon_cond} (and  \eq{double_horizon_cond2}) supplemented with the
  condition that $f$ has no zeroes has  been established in \cite{LRS} as a necessary and sufficient condition
on the data on a bifurcate horizon to generate a near-horizon geometry (which contains a non-degenerate MKH). The solvability of this equation is analyzed in $3+1$ dimensions
if the bifurcation surface $S$ is topologically a 2-sphere in \cite{ABL,LP,LRS}, and in particular in \cite{DKLS} it is shown that in $3+1$-dimensions, and if the bifurcation surface $S$ is compact with positive genus, the only solutions $(\gamma, \bm{\torsion})$ to  \eq{special_gauge_master}
satisfy ${}^{\gamma}R=2\Lambda$ and $\bm{\torsion}=0$.  Note, though, that in our setting even for a compact $S$, if the solution to \eq{double_horizon_cond} has zeros
the gauge where \eq{special_gauge_master} holds  can  in general be realized only on non-compact subsets of $S$. We devote Section
\ref{sec_NHG} to discuss in depth the connection of the KID equations with near horizon geometries and, in particular, put forward a generalized definition of near-horizon geometries admitting zeros of the degenerate Killing vector.
}
\end{remark}

The restriction of the emerging Killing vector to the bifurcate horizon is computed by integrating the
ODEs \eq{ODE_KID1}, \eq{ODE_KID2} and \eq{ODE_KID3}
on $\mathcal{H}_1^+$, and the corresponding ones on $\mathcal{H}_2^+$.
If $\eta$ has $\mathcal{H}^+_1$ as Killing horizon we have
\begin{eqnarray}
\eta_u & \overset{\mathcal{H}^+_1}{=}&0
\,,
\label{candidate1}
\\
 \eta_A& \overset{\mathcal{H}^+_1}{=}&0
\,,
\\
\eta^u& \overset{\mathcal{H}^+_1}{=}& f + \kappa_{\eta} u
\,, \quad  \kappa_{\eta} = \mathrm{const.}
\label{candidate3}
\\
\eta_v& \overset{\mathcal{H}^+_2}{=}&  f
\,,
\\
\eta_A& \overset{\mathcal{H}^+_2}{=}& -v (D_A-2\torsion_A) f
\,.
\label{candidate5}
\end{eqnarray}
The  ODE for $\eta^v|_{\mathcal{H}^+_2}$ (equivalently $\eta_u|_{\mathcal{H}^+_2}$)
yields a more complicated solution, so let us just  mention here that the initial data are  given by $\eta^v\overset{S}{=}0$ \eq{cond_double_horizon2}  and
$\partial_v\eta^{v}\overset{S}{=}- \kappa_{\eta}- f \partial_ug_{uv}$
\eq{tran_deriv2}.

As the notation already suggests, the constant $ \kappa_{\eta}$ can be identified with  the surface gravity of $\mathcal{H}_1^+$ associated to  $\eta$.
Adding to the candidate field $\eta$ an appropriate multiple of $\xi$, one may always assume $ \kappa_{\eta}=0$, i.e.\ that
the associated Killing horizon is degenerate.
In this sense a solution to \eq{double_horizon_cond} uniquely defines a Killing vector field $\eta$ for which   $\mathcal{H}_1^+$ is a \emph{degenerate} horizon.

We also observe that
%
\begin{eqnarray*}
[\xi,\eta]^{\mu}\overset{S}{=} -\eta^u\partial_u\xi^u\delta^{\mu}{}_u\overset{S}{=}- \kappa_{\xi} f  \delta^{\mu}{}_u \ne 0
\,,
\end{eqnarray*}
i.e.\ apart from multiples of $\xi$ no Killing vector which has $\mathcal{H}_1$ or $\mathcal{H}_2$ as Killing horizon can commute with $\xi$.
This is in accordance with Theorem~\ref{thm_MKH2}.

As a straightforward consequence of Theorem~\ref{prop_KID_equations} and Corollary~\ref{cor_multiple_horizons} we have:
\begin{corollary}
Let  $( M ,g)$ be an $(n+1)$-dimensional $\Lambda$-vacuum spacetime, $n\geq 3$, with at least $\frac{1}{2}n(n-1)+2$
Killing vectors. Then any bifurcate Killing horizon contains a non-degenerate MKH.
\end{corollary}
\begin{proof}
The KID equation \eq{KID1} admits at most $\frac{1}{2}n(n-1)$ independent solutions so \eq{KID2}-\eq{KID3} must admit at least one non-trivial solution.
\end{proof}

Another useful consequence of the existence result is the following lemma, to be used later.
\begin{lemma}
  \label{density}
  Let $(S,\gamma)$ by an $(n-1)$-dimensional Riemannian  manifold ($n \geq 3$)
  endowed with a one-form
  $\bm{\torsion}$. Let $f : S \rightarrow \mathbb{R}$ be
  a non-identically zero solution of
  (\ref{double_horizon_cond})
  or (\ref{double_horizon_cond2}). Then $f$ is non-zero on a dense subset of $S$.
\end{lemma}
\begin{proof}
  We consider only (\ref{double_horizon_cond}), the other case is similar.
Let $(M,g)$ be a $\Lambda$-vacuum spacetime with bifurcate Killing horizon
  generated by the data $(S,\gamma,\bm{\torsion})$ and $\eta$ the Killing vector
  on $D^+ (\H_1^+ \cup \H_2^+)$ satisfying, in corresponding null
  coordinates, (\ref{candidate1})-(\ref{candidate5}) with
  $\kappa_{\eta}=0$. Assume that the set $\{ p \in S; f (p) =0\}$  has a
  non-empty interior $S^{(0)}$. Then by Theorem \ref{theorem_fund_form}, $\eta$ is identically
  zero on some neighbourhood of $S$, which can only happen for a Killing if it
  is zero everywhere and hence $f$ was in fact identically zero on $S$.
  \end{proof}

\subsection{Non-degenerate MKHs of order $m\geq 3$}
\label{sec_non_deg3}

Let us assume that a $\Lambda$-vacuum spacetime admits  a bifurcate Killing horizon such that $\mathcal{H}^+_1$
is a multiple Killing horizon of order $m\geq 3$
and let
$\bm{\torsion_0}$ be  its torsion one-form. Assume that the bifurcation surface is connected.
Then \eq{double_horizon_cond} (with $\bm{\torsion}$ replaced
by $\bm{\torsion_0}$)
admits  at least two independent non-trivial solutions.
Select one such solution $f^{(1)}$ and let $\widetilde{S}$ be the
set where $f^{(1)}$ is non-zero.  By
Lemma \ref{density}, $\widetilde{S}$ is dense on $S$. On this set we can realize
a gauge where $ f ^{(1)}=1$. Throughout this section we denote by
$\bm{\torsion}$ the corresponding torsion-one-form on $\widetilde{S}$ where
(\ref{special_gauge_master}) holds.


The second solution will be denoted by $ f $. 
It cannot be constant on open sets (otherwise it is  proportional
to $f^{(1)}$ on this open set, hence everywhere by
Lemma \ref{density} and we would contradict the linear independence)
and in our current gauge solves the
master equation \eq{Master_Equation},
\begin{equation}
\Big(D_{A} D_{B}
-2  \torsion_{(A}D_{B)} \Big) f
= 0
\,.
\label{master}
\end{equation}
As a solution is uniquely determined by the values $ f $ and $\mathrm{d} f $  at one point, 
and all constant functions
solve (\ref{master}) ,
we have
\begin{equation}
N:= \gamma^{\sharp}(\mathrm{d}f ,\mathrm{d}f)  >0
\,
\label{ineq_F}
\end{equation}
everywhere on $\widetilde{S}$.
Equation \eq{master} has the following first integrability condition 
(see
\cite[Equation (61)]{mps})
%
\begin{equation}
-{}^{\gamma} R_{ABCD}D^D f  +2D_C f D_{[A}\torsion_{B]}  - 2\torsion_C \torsion_{[A}D_{B]} f - 2 D_{[A} f D_{B]}\torsion_C
=0
\,.
\label{integrability}
\end{equation}
 We contract with  $D^C f $ and $\gamma^{BC}$, respectively,
\begin{eqnarray}
 2ND_{[A}\torsion_{B]}
  +2s^CD_C fD_{[A} f \torsion_{B]}
-D_{A} fD^C fD_{[B}\torsion_{C]}
+D_{B} f D^C f D_{[A}\torsion_{C]}
\nonumber
&&
\\
-D_{A} fD^C fD_{(B}\torsion_{C)}
+ D_{B} f D^C fD_{(A}\torsion_{C)}
&=&0
\,,\phantom{xxx}
\label{inte1}
\\
3D^B f D_{[A}\torsion_{B]}
-  (D_{B}s^B- |s|^2)D_A f
+D^{B} f (D_{(A}\torsion_{B)} -\torsion_{A}\torsion_B)
+{}^{\gamma} R_{AB} D^B f
&=&0
\,,
\label{inte2}
\end{eqnarray}
and  use the gauge condition \eq{special_gauge_master},
\begin{eqnarray*}
 2ND_{[A}\torsion_{B]}
-D_{A} f D^C f D_{[B}\torsion_{C]}
+ D_{B} fD^C fD_{[A}\torsion_{C]}
-{}^{\gamma} R_{C[A}D_{B]} fD^C f
&=&0
\,,
\\
3D^B f D_{[A}\torsion_{B]}
+\frac{1}{2} \Big({}^{\gamma} R-2(n-2)\frac{\Lambda}{n-1}\Big)D_A f
+\frac{1}{2}  D^{B} f {}^{\gamma} R_{AB}
&=&0
\,.
\end{eqnarray*}
We insert the second equation into the first one,
\begin{equation}
 ND_{[A}\torsion_{B]}
- 2D_{A} f D^C fD_{[B}\torsion_{C]}
+ 2D_{B} f D^C f D_{[A}\torsion_{C]}
=0
\,,
\label{sigma_closed}
\end{equation}
and  apply $D^B f $,
\begin{equation*}
 ND^B f D_{[A}\torsion_{B]}
=0
\quad\Longrightarrow \quad D^B fD_{[A}\torsion_{B]}
=0
\,.
\end{equation*}
%
From this we deduce via \eq{sigma_closed} and \eq{ineq_F} that the torsion one-form  $\bm{\torsion}$
must  be closed on $\widetilde S$. Thus, the same is true for the
original torsion $\bm{\torsion_0}$
(as they differ  by an exact form).
Unlike $\bm{\torsion}$,
$\bm{\torsion_0}$ is smooth on the whole of $S$, so it must be closed everywhere:

\begin{lemma}
\label{lem_closed}
 Consider an $(n+1)$-dimensional  $\Lambda$-vacuum spacetime which  admits  a bifurcate Killing horizon such that $\mathcal{H}_1$
 is a multiple Killing horizon of order $m\geq 3$. Then its torsion 1-form
is closed.
\end{lemma}

\begin{remark}
  {\rm 
  In the setup of this lemma (so that in particular
    $\bm{\torsion}$ is closed), consider the domain $\tilde{S}$
    and define $h: \widetilde{S} \rightarrow \mathbb{R}$ by
\begin{equation}
h := 2|\bm{\torsion}|^2-D_A \torsion^A +\frac{2}{n-1}\Lambda=\frac{1}{2}\big({}^{\gamma} R-\frac{2(n-3)}{n-1}\Lambda\big) +|\bm{\torsion}|^2
\,.
\label{dfn_h}
\end{equation}
%
Applying $D^A$ to   \eq{special_gauge_master} and using the contracted Bianchi identity for $\gamma$, it follows
that (compare \cite{KuLu} where such an  equation appears in the context of vacuum
near-horizon geometries)
\begin{equation}
D_A h - 2h\torsion_A=0
\label{equation_h}
\end{equation}
which obviously implies
\begin{align}
  D_A h D^A f = 2 h \torsion_A D^A f.
  \label{equationDhDf}
\end{align}
%
%
Inserting  \eq{dfn_h} and \eq{special_gauge_master} into  \eq{inte2}
one obtains
\begin{equation}
h D_A f
=
D_{A} (\torsion^B D_{B} f )
 -2\torsion_{A}\torsion_BD^{B} f
\,.
\label{add_eqn}
\end{equation}
Equations (\ref{master}), (\ref{equationDhDf}) and \ref{add_eqn})
arise naturally also in a different but related context, and several consequences of them have been obtained in \cite{mps2}. From those results  one can extract the following lemma.
\begin{lemma}
  \label{Lemmamps2}
  Let $(\widetilde{S},\gamma)$ be an $(n-1)$-dimensional
  Riemannian manifold ($n \geq 3$) with a one-form $\bm{\torsion}$
  satisfying (\ref{special_gauge_master}).
  Let $p+1$  (necessarily $\leq n$)  be the dimension of the space of solutions  of (\ref{master}).
  If $p \geq 1$, then
  \begin{itemize}
  \item[(i)] $\torsion$ is exact,
  \item[(ii)] locally, $(\widetilde{S},\gamma)$ is  a warped product
    $\widetilde S = V\times \Sigma$, $\gamma= \overline{\gamma}
    + \Omega \mathring \gamma$, $\Omega : V \rightarrow \mathbb{R}$,
  where $(\Sigma,\mathring\gamma)$ is
a $p$-dimensional maximally symmetric Riemannian manifold.
     \end{itemize}
\end{lemma}
  We emphasize that, in the context of multiple Killing horizons of order
  $m \geq 3$ where this lemma automatically applies,
    the globally defined torsion one-form
$\bm{\torsion}_0$ is in general not exact.
}
\end{remark}


\begin{theorem}
\label{thm_manyKVs}
Let $( M , g)$ be an $(n+1)$-dimensional $\Lambda$-vacuum spacetime which admits a bifurcate Killing horizon
with one of its Killing horizons being a non-degenerate MKH of order $m\geq 3$. Then the spacetime admits (locally) at least $\frac{1}{2}m(m+1)$ Killing vectors.
\end{theorem}

\begin{proof}
Without restriction assume that $\mathcal{H}_1^+$ is non-degenerate MKH of order $m\geq 3$.
We consider the initial data on the bifurcation surface which generate this spacetime.
Moreover, we restrict attention to the subset $\widetilde S\subset S$ where the gauge \eq{special_gauge_master}
can be realized.
In particular \eq{master}, \eq{equationDhDf} and \eq{add_eqn} hold and
the space of solutions of \eq{master} is $m-1$. By item (ii) of
Lemma \ref{Lemmamps2}, $(\widetilde  S, \gamma|_{\widetilde S})$
admits at least $\frac{1}{2}(m-2)(m-1)$ (local) Killing vectors. Moreover, since
$\bm{\torsion}$ is exact, equation (\ref{KID3}) restricted to
$\widetilde{S}$ admits the same number of
independent solutions as (\ref{master}). Applying
Theorem~\ref{prop_KID_equations}, the data on $\widetilde S$, and therefore on $S$,
generate a spacetime with
\begin{align*}
\frac{1}{2}(m-2)(m-1) + 2 (m-1) + 1 = \frac{1}{2} m(m+1)
\end{align*}
locally defined Killing vectors, as claimed.
\end{proof}

\subsection{Vanishing torsion one-form}
\label{sec_vanish_torsion}

A particular case of relevance for characteristic initial data corresponding
to a bifurcate horizon  is  the case when the torsion one-form $\torsion$ is
exact and  $(S,\gamma)$ is Einstein, i.e.
\begin{align}
  {}^{\gamma} R_{AB}= \frac{{}^{\gamma}  R}{n-1} \gamma_{AB}.
\label{weaker_assumpt_data}
  \end{align}
Note that for $n \geq 4$ the second Bianchi identity then  implies  ${}^{\gamma}  R=\mathrm{const.}$, while for $n=2,3$ (\ref{weaker_assumpt_data}) imposes no restriction whatsoever on $\gamma$.  The case  $n=2$
is special since all KID equations become ODEs, so we assume $n \geq 3$.  We also assume that $S$ is connected.

We want to analyze the space of solutions of the KID equations for such data.
When $\bm{\torsion}$ is exact, there is a global gauge where $\bm{\torsion}=0$,
which we assume from on. Note that when $\bm{\torsion}$ is merely closed,
this  gauge can also be imposed locally, so all results below of a local nature
also hold in this case.

The KID equation to be solved is
(cf.\ \eq{double_horizon_cond} and \eq{double_horizon_cond2})
\begin{equation}
\Big(D_{A}D_{B}
-\frac{{}^{\gamma} R}{2(n-1)} \gamma_{AB}+\frac{\Lambda}{n-1} \gamma_{AB} \Big) f
= 0
\,.
\label{max_sym_KID}
\end{equation}
We split \eq{max_sym_KID} into trace and  trace-free part,
\begin{eqnarray}
(D_{A} D_{B} f )_{\mathrm{tf}}&=& 0
\,,
\label{khs_eqn1}
\\
\Big(\Delta_{\gamma}-\frac{1}{2}{}^{\gamma} R +\Lambda) \Big) f &=& 0
\,.
\label{khs_eqn2}
\end{eqnarray}
We determine the divergence of the first equation,
\begin{equation}
D_{A}\Delta_{\gamma}  f
+\frac{1}{n-2}{}^{\gamma} RD_{A} f
=0
\,.
\label{divergence_eqn}
\end{equation}
Inserting \eq{khs_eqn2}
we obtain
\begin{equation}
\Big(\frac{n}{n-2}{}^{\gamma} R-2\Lambda\Big)D_{A}  f
+ f D_{A}{}^{\gamma}  R
=0
\,.
\label{ODE_RY}
\end{equation}
%
When the first parenthesis is not identically zero, which we write as
\begin{equation}
  {}^{\gamma}  R
  \not \equiv \frac{2(n-2)}{n}\Lambda
\,,
\label{ineq_rR}
\end{equation}
the equation can be integrated, the general solution being
\begin{equation}
  f =C \left | \frac{n}{n-2}{}^{\gamma}   R-2\Lambda \right |^{\frac{2-n}{n}}  \,, \quad  C = \mathrm{const.}
 \label{solf}
\end{equation}
We insert it into \eq{max_sym_KID},
\begin{align}
&\frac{2(n-1)}{n-2}D_{A} {}^{\gamma}  RD_{B}{}^{\gamma}  R
-\Big( \frac{n}{n-2}{}^{\gamma}  R-2\Lambda \Big)D_{A}D_{B}{}^{\gamma}  R
\nonumber
\\
&\qquad\qquad  -\frac{1}{n-1}\Big(\frac{{}^{\gamma}  R}{2} -\Lambda \Big) \Big( \frac{n}{n-2}{}^{\gamma}  R-2\Lambda \Big)^2\gamma_{AB}
= 0
\,.
\label{2-order_KH_cond}
\end{align}
Summarizing, the following result has been proved.
\begin{lemma}
\label{lem_multiple}
Consider data $(S,\gamma,\bm{\torsion})$ generating a bifurcate Killing horizon.
Assume that $S$ is connected of dimension at least two,
$\bm{\torsion}$ is  exact
and \eq{weaker_assumpt_data}, \eq{ineq_rR} hold. Then
a MKH which belongs to the bifurcate Killing horizon can be at most of order 2.
It is of order 2 if and only if \eq{2-order_KH_cond} holds and then the solution of the KID equation in the gauge $\bm{\torsion}=0$ is (\ref{solf}).
\end{lemma}

\begin{remark}
\label{remark_n4-case}
{\rm
If  $n\geq 4$
equation \eq{2-order_KH_cond} implies by \eq{weaker_assumpt_data}  and \eq{ineq_rR}   that there there will be an MKH of order 2 if and only if
${}^{\gamma}  R=2\Lambda$.  In that case $f$ is constant on $S$.
}
\end{remark}

Let us now restrict the data further by imposing
${}^{\gamma}R=\mathrm{const.}$ This holds in particular if  $(S,\gamma)$ is maximally symmetric.
 which we recall to be equivalent to
\begin{equation}
  {}^{\gamma} R_{ABCD} =\frac{2\, \, {}^{\gamma}  R}{(n-1)(n-2)} \gamma_{A[C}\gamma_{D]B}\,,
\end{equation}
with
${}^{\gamma}  R =\mathrm{const.}$ We compute the  divergence of \eq{divergence_eqn},
\begin{equation}
\Delta_{\gamma} \Big(\Delta_{\gamma}+\frac{{}^{\gamma}  R}{n-2}\Big) f = 0
\,,
\end{equation}
and plug in \eq{khs_eqn2}
\begin{equation}
\Big({}^{\gamma}  R-2\Lambda\Big) \Big(\frac{n}{n-2}{}^{\gamma}  R-2\Lambda\Big) f = 0
\,,
\end{equation}
which requires
\begin{equation}
{}^{\gamma}  R=2\Lambda\quad \text{or} \quad
{}^{\gamma}  R
=\frac{2(n-2)}{n}\Lambda
\,.
\label{some_nec_conds}
\end{equation}

\begin{remark}
{\rm
If $n\geq 4$ data of the form \eq{weaker_assumpt_data} always have constant scalar curvature so the same conclusion can be drawn:
A necessary  condition for the existence of  a multiple Killing horizon is \eq{some_nec_conds}, and ${}^{\gamma} R=2\Lambda$ is, by Remark~\ref{remark_n4-case}, also sufficient.
}
\end{remark}

\subsubsection{Case  ${}^{\gamma}   R=2\Lambda$}
Let us consider the case where ${}^{\gamma}   R=2\Lambda$ first. It is different from the second  case only for $\Lambda \ne 0$, which we assume here.
Then Lemma~\ref{lem_multiple} applies and $ f =\mathrm{const.}\ne 0$ is the only candidate solution, and, indeed, solves
\eq{max_sym_KID}. It then follows from Theorem
\ref{prop_KID_equations} that the number of Killings of the vacuum solution generated by the data is $3 + k$, where $k$ is the dimension of the Killing algebra of $(S,\gamma)$. Thus,
\begin{lemma}
  \label{lem_Nariai}
  With the same assumptions as in Lemma \ref{lem_multiple}, impose further
  ${}^{\gamma}  R  = 2 \Lambda \neq 0$.
  Let $k \geq 0$ be the number of
  independent Killing vectors of $(S,\gamma)$. Then the
  vacuum spacetime generated by this data admits   $k+3$ Killing vectors and
  both horizons $\mathcal{H}^{+}_a$, $a=1,2$, which belong to the bifurcate Killing horizon  are  multiple Killing horizons of order 2.
\end{lemma}

 It follows from the comment just after Theorem \ref{theorem_fund_form} that the second Killing vector with $\mathcal{H}^{+}_1$ as Killing horizon must be different from the second Killing vector with  Killing horizon $\mathcal{H}^{+}_2$.

 Lemma~\ref{lem_Nariai} confirms the existence (for $\Lambda\ne0$) of non-maximally symmetric spacetimes with non-degenerate multiple Killing horizons, as was already shown in 4 dimensions in \cite[section 4.2]{mps}.

Since maximally symmetric spacetimes of dimension $n-1$ admit
$\frac{1}{2}n(n-1)$ independent Killing we also have from Lemma \ref{lem_Nariai}
\begin{corollary}
  Let $(S,\gamma)$ be maximally symmetric, connected and of dimension
  at least two and assume  ${}^{\gamma}  R  = 2 \Lambda \neq 0$.
  Then the spacetime generated by the data
  $(S,\gamma,\bm{\torsion} = d \psi)$ has $\frac{1}{2} n (n-1)  +3$ Killing vectors and  both horizons $\mathcal{H}^{+}_a$, $a=1,2$
  are  multiple Killing horizons of order 2.
  \end{corollary}

\begin{remark}
{\rm
The bifurcation surface of the (Anti-)Nariai spacetime
\cite{Exact}
 is, depending on the sign of $\Lambda$,   a round sphere or hyperbolic space with Ricci scalar ${}^{\gamma}  R=2\Lambda$,  and has vanishing torsion, i.e.\ by uniqueness of solution to the characteristic Cauchy problem it must be the solution predicted by  the corollary.

For $\Lambda>0$ there is another way of seeing this: One may start with a round sphere as  bifurcation surface. Then the emerging spacetime will be spherically symmetric and  by Birkhoff's theorem (cf.\ e.g.\ \cite{schmidt})
the emerging spacetime must locally  be isometric to the (Anti-)Nariai solution  (Schwarzschild-de Sitter can be excluded since it has only $\frac{1}{2}n(n-1)+1$ Killing vectors).
}
\end{remark}

\subsubsection{Case  ${}^{\gamma}  R
=\frac{2(n-2)}{n}\Lambda $}

Let us devote attention now to the second case where
\begin{equation}
{}^{\gamma}  R
=\frac{2(n-2)}{n}\Lambda
\,.
\label{cond_scalar_curv}
\end{equation}
%
Maximally symmetric vacuum spacetimes have vanishing Weyl tensor, so the Gauss identity  on $S$ gives
$0= C_{ABCD}\overset{s}{=}{}^{\gamma}  R_{ABCD} -\frac{4}{n(n-1)}\Lambda \gamma_{A[C}\gamma_{D]B}$, whence
\eq{cond_scalar_curv} is necessary for the spacetime to be maximally symmetric.

Now  \eq{max_sym_KID}  becomes
\begin{eqnarray}
(D_{A}D_{B} f )_{\mathrm{tf}}&=& 0
\,,
\label{max_sym_KID1}
\\
\Big(\Delta_{\gamma}+\frac{2}{n}\Lambda  \Big) f &=& 0
\,.
\label{max_sym_KID2}
\end{eqnarray}
%
%
%
 Let us assume that $(S, \gamma)$ is maximally symmetric and, in addition, simply connected, connected and complete. Then, depending on the sign of the cosmological constant,
the bifurcation surface can be identified with Euclidean space, hyperbolic space or a round sphere.
In fact, if we require \eq{max_sym_KID1}-\eq{max_sym_KID2} to admit a solution,  the  simply connectedness-requirement is not needed
if the cosmological constant is positive \cite{obata,tashiro}.

\paragraph{Vanishing cosmological constant}
If $\Lambda=0$ then $(S, \gamma)$ needs to be the Euclidean space and in that case $n$ independent solutions are given, in standard
coordinates $(x^A)$,  by $ f =1$ and $ f =x^A$, $A=3,\dots, n+1$.

\paragraph{Positive cosmological constant}
If $\Lambda>0$ let us rescale the metric  such that $\Lambda=n(n-1)/2$.
Then ${}^{\gamma}  R=(n-1)(n-2)$, while $(S, \gamma)$ is the standard $(n-1)$-sphere.
In that case $n$ independent solutions to \eq{max_sym_KID2} are given by the $\ell=1$-spherical harmonics $Y_{\ell}$, $\Delta_{\gamma} Y_{\ell}=-\ell(\ell+n-2)Y_{\ell}$.
Since the gradient of each function with  $\Delta_{\gamma} Y_{\ell}=-(n-1)Y_{\ell}$ is a conformal Killing one-form \cite{obata}, all $Y_{\ell}$'s satisfy the full system
\eq{max_sym_KID1}-\eq{max_sym_KID2}.

\paragraph{Negative cosmological constant}
If $\Lambda<0$ we rescale the metric  such that $\Lambda=-n(n-1)/2$.
Then ${}^{\gamma}  R=-(n-1)(n-2)$, while $(S, \gamma)$ is the standard hyperbolic space.
By analytic continuation of the spherical harmonics  of degree $\ell$ on an $n-1$-dimensional standard sphere to the $n-1$-dimensional standard  hyperbolic space, both regarded
as subsets of the complex unit sphere in $\mathbb{C}^n$, one obtains spherical   harmonics  of degree $\ell$  on the hyperbolic space \cite{strichartz}.
These spherical harmonics are eigenfunctions of $\Delta_{\gamma}$ with eigenvalue $\ell(\ell+ n-2)$.
In particular the $\ell=1$-spherical  harmonics  provide $n$  independent solutions which solve  \eq{max_sym_KID2}.

When the conditions on $(S,\gamma)$ being connected, simply connected or
complete are dropped, the KID equations (\ref{max_sym_KID1})-(\ref{max_sym_KID2}) still admit $n$ linearly independent local
solutions.  Thus, the following result holds (note that since the
statement is local we may assume that $\bm{\torsion}$ is merely closed).

\begin{lemma}
  \label{MaxSym}
  Let $(S,\gamma)$ be maximally symmetric and satisfying
  (\ref{cond_scalar_curv})  and $\bm{\torsion}$
  be closed. Assume $n \geq 3$ and let $(M,g)$ be a
  spacetime generated by the characteristic initial value problem corresponding to a bifurcate Killing horizon with data $(S,\gamma,\bm{\torsion})$.
 Then any point $q\in S \subset M$ admits
  a spacetime neighbourhood $U_q \subset M$  with $\frac{1}{2} (n+1)(n+2)$ Killing vectors. In particular $(M,g)$
is maximally symmetric in some neighbourhood of $S$.
\end{lemma}

One may also consider whether the converse is true, i.e.\ given a
spacetime with a bifurcate Killing horizon with bifurcation surface $S$
and maximally symmetric near $S$, whether  the data $(S,\gamma,\bm{\torsion})$ must
be as in Lemma \ref{MaxSym}.
According to Remark~\ref{rem_max_sym} (iii) the bifurcation surface $(S, \gamma)$ needs to be maximally symmetric, i.e.\
admit $\frac{1}{2}n(n-1)$  independent (local) Killing vectors
 $\ol \zeta^{(b)}$, $b\in\{1,\dots,n(n-1)/2\}$.
Each of these  Killing vectors need to extend (also locally) to a spacetime symmetry which requires, by Theorem~\ref{theorem_fund_form}, that
$$
\mcL_{\ol \zeta^{(b)}}D_{[A}\torsion_{B]} = 0 \enspace \forall\, b
\quad \Longrightarrow \quad  D_{[A}\torsion_{B]} = 0
\,.
$$
i.e.  $\bm{\torsion}$ needs to be closed.
%
Finally, as already mentioned the Gauss identity on $S$ forces
(\ref{cond_scalar_curv}). Thus, indeed the converse of Lemma \ref{MaxSym}
holds true.

Combining this converse with the uniqueness of solutions
of  the characteristic initial value problem near the initial hypersurface we conclude
\begin{lemma}
  Let $(S,\gamma)$ be maximally symmetric, complete, connected and
  simply connected and $\bm{\torsion}$ exact. Let $(M,g)$ be a
  spacetime generated by the characteristic initial value problem corresponding to a bifurcate Killing horizon with data $(S,\gamma,\bm{\torsion})$. Then,
  $(M,g)$ is isometric to a portion of de Sitter ($\Lambda >0$),
  Minkowski ($\Lambda=0$) or
  anti-de Sitter ($\Lambda <0$).
\end{lemma}

\vspace{1cm}

\subsection{Non-degenerate MKHs  of order $m=2$ for  $\Lambda=0$}
\label{sec_order2MKH}

The (Anti-)Nariai solution provides an example of a $\Lambda\lessgtr 0$-vacuum spacetime which admits a non-degenerate MKH of order 2.
We therefore aim to construct characteristic initial data which generate $\Lambda=0$-vacuum spacetimes with  non-degenerate MKH of order 2.
For this we  consider, in $n+1$ dimensions, initial data with vanishing torsion 1-form.
Then \eq{double_horizon_cond} becomes
\begin{equation}
\Big(D_{A}D_{B}
-\frac{1}{2}{}^{\gamma} R_{AB} \Big) f
= 0
\,.
\label{double_horizon_cond2b}
\end{equation}
We want to find a metric $\gamma$ for which these equations admit  precisely one  non-trivial solution.
For this we assume $n \geq 3$ (if $n=2$, then
\eq{double_horizon_cond2b} is a second order linear ODE which always has two independent solutions)  and that
$(S,\gamma)$ is a warped product metric with locally flat $n-2$ dimensional fibers,  so that the metric can be written in the local form
\begin{equation}
\gamma= \mathrm{d}x^2 + \R^2(x)\delta
\,,\quad \delta=\sum_{P=1}^{n-2}(\mathrm{d}y^P)^2
\,, \quad n\geq 3
\,. \\
\label{MKH2_metric}
\end{equation}
Then
\begin{align}
  {}^{\gamma} R_{xx} =  - (n-2)\frac{\partial_x^2\R}{\R}
\,,
\quad
  {}^{\gamma} R_{PQ} =
  - \left ( (n-3) (\partial_x \R)^2 + \R \partial_x^2 \R
  \right )   \delta_{PQ}
\,,
\\
 {}^{\gamma} R_{xP}=0
  \,, \quad {}^{\gamma} R=-(n-2) \left ( (n-3) \frac{(\partial_x \R)^2}{\R^2}
  - 2 \partial_x^2 \R \right )
\,,
\end{align}
and  \eq{double_horizon_cond2b}
 becomes
\begin{align}
  \partial_x^2f + \frac{(n-2)}{2} \frac{\partial_x^2\R}{\R}   f  =& 0\,, \label{Eqxx} \\
  \partial_P\partial_{Q}f + \R \partial_x \R \delta_{PQ}\partial_xf
  + \frac{1}{2} \left ( (n-3) (\partial_x \R)^2 + \R \partial_x^2 \R
  \right )   \delta_{PQ} f =& 0\,, \label{EqPQ} \\
  \Big( \partial_x-\frac{\partial_x\R}{\R }\Big)\partial_Pf =& 0\,. \label{EqPx}
\end{align}
We assume that the metric $\gamma$ is not flat, as otherwise
\eq{double_horizon_cond2b} clearly admits $n$ (hence more than one)
linearly independent
solutions. Thus,  $\partial_x \R$ is not identically zero and when
$n =3$ also $\partial_x^2 \R$ is not identically zero. When $n >3$ it is easy to show that if $\R = \R_0 x + \R_1$ with
$\R_0$ a non-zero constant, the system (\ref{Eqxx})-(\ref{EqPx})
admits no non-trivial solutions. So, we may assume that neither
$\partial_x \R$ nor $\partial_{x}^2\R$ vanish identically.

Equation (\ref{EqPx}) integrates to $f = a(y) \R + b(x)$, which inserted into
\eq{Eqxx} yields
\begin{align*}
  \frac{n \partial_x^2 \R}{2} a(y) + \partial_{x}^2 b
  + \frac{n-2}{2} \frac{\partial_x^2 \R}{\R} b =0.
\end{align*}
 Thus, $a(y)$ is necessarily constant and $f$ only depends on $x$. Thus the system to solve
  is
  \begin{align}
  \partial_x^2f + \frac{(n-2)}{2} \frac{\partial_x^2\R}{\R}   f  =& 0\,, \label{Eqxx2} \\
  \R \partial_x \R \partial_xf
  + \frac{1}{2} \left ( (n-3) (\partial_x \R)^2 + \R \partial_x^2 \R
  \right )  f =& 0\,. \label{EqPQ2}
\end{align}
This is a system of two ODEs for two functions. We aim at finding its general solution. Replacing $x \rightarrow -x$ if necessary we may assume that
$\partial_x \R >0$. Equation \eq{EqPQ2} can be integrated to
\begin{align*}
  f = \frac{f_0 \R^{- \frac{n-3}{2}}}{\sqrt{\partial_x \R}}
\end{align*}
where $f_0$ is a constant. Inserting into the other equation gives a rather
involved third order ODE for $\R(x)$. It is more convenient to perform
the change of variable
\begin{align}
  dx = \frac{z dz}{w(z)}
  \label{changevar}
\end{align}
and use the freedom introduced by $w(z)$ to impose the equation
\begin{align}
  \frac{d\R}{dz} = \frac{\R}{z w}.
  \label{EqR}
\end{align}
Note that in terms of this variable, $f=  f_0 z \R^{-\frac{n-2}{2}}$. A straightforward computation shows
that equation \eq{Eqxx2} becomes
\begin{align}
  \label{ODEfinal}
  z w \frac{dw}{dz} = w^2 + w (n-2) - \frac{n(n-2)}{4} := P^{(n)}(w).
\end{align}
This is a separable equation which can be solved explicitly. Note that
the second order polynomial $P^{(n)}(w)$ vanishes at two real and non-zero values
\begin{align*}
  w_{\pm} := \frac{ - (n-2) \pm \sqrt{ 2 (n-1)(n-2)}}{2},
\end{align*}
so that  $w(z) := w_{\pm}$ solve the ODE.  When $w(z)$ is not constant,
$P^{(n)}(w)$
is not identically zero and we get
\begin{align*}
  z =  \exp\left ( {{\int \frac{dz}{z}} } \right ) =
  \exp \left ( { {\int\frac{w dw}{P^{(n)}(w)}}} \right )=
  \frac{z_0 |w-w_{+}|^{\frac{w_+}{w_+ - w_-}}}{|w-w_{-}|^{\frac{w_-}{w_+ - w_-}}},
  \end{align*}
where $z_0$ is a constant. We elaborate further each case. When $w := w_{\pm}$, equation
(\ref{changevar}) integrates to
\begin{align*}
  z^2 = 2 w_{\pm} (x-x_0).
\end{align*}
The constant $x_0$ may be set to zero by shifting $x$. Then
(\ref{EqR}) gives
\begin{align*}
  \R(z) = c_0 |z|^{\frac{1}{w_{\pm}}} := c_0 |z|^{p_{\pm}}
    = \widetilde{c}_0 |x|^{\frac{p_{\pm}}{2}}
\end{align*}
where $c_0, \widetilde{c}_0$ are non-zero constants and
\begin{align*}
  p_{\pm} := \frac{2}{n}\left (1\pm\sqrt{\frac{2(n-1)}{n-2}}\right )\,.
\end{align*}
The constant $\widetilde{c}_0$ can be absorbed into the coordinates $\{y^P\}$
and we can also assume $x>0$ without loss of generality. So, the metric
and the function $f$ takes the form (redefining $f_0$ conveniently)
\begin{align*}
  \gamma = dx^2 + x^{p_{\pm}} \delta, \quad \quad
  f = f_0 x^{\frac{2 -(n-2) p_{\pm}}{4}}.
  \end{align*}
Observe that $p_{\pm} = 2$ only for the plus sign and $n=3$ and this case
has to be excluded as it corresponds to a flat metric.  Thus, for $n=3$ only $p_{-} = -\frac{2}{3}$ survives.

When $w(z)$ is not constant, we may use $w$ as coordinate. The equation for $\R(w)$ is, from (\ref{EqR}) and (\ref{ODEfinal})
\begin{align*}
  \frac{d\R}{dw} = \frac{\R}{(w-w_{+})(w-w_-)}
  \quad \quad \Longrightarrow \quad \quad \R(w)= \R_0 \left | \frac{w-w_{+}}{w-w_{-}} \right |^{\frac{1}{w_{+} - w_{-}}}, \quad \R_0 \in \mathbb{R}^+
\end{align*}
so that the metric  is, after absorbing a multiplicative constant
in $\{y^P\}$ and  redefining a suitable constant $C>0$,
\begin{align}
  \gamma = C   \left | \frac{w-w_+}{w - w_{-}} \right |^{\frac{2(w_+ + w_-)}{w_+ - w_-}}  dw^2 +
\left | \frac{w-w_+}{w - w_{-}} \right |^{\frac{2}{w_+ - w_-}}
  \delta.
  \label{metricsol}
\end{align}
In this case the solution of
(\ref{double_horizon_cond2b}) is
\begin{align*}
  f := f_0 \frac{\left | w -w_+ \right |^{\frac{3 w_+ + w_-}{2(w_+ - w_-)}}}
  {{\left | w -w_- \right |^{\frac{w_+ + 3 w_-}{2(w_+ - w_-)}}}}
  \end{align*}
  Summarizing:
  \begin{lemma}
\label{kasner}
Consider data $(S, \gamma, \bm{\torsion}=0)$, with
$(S,\gamma)$ a warped product  with locally flat $(n-2)$-dimensional fibers
 for the characteristic initial value problem for the $(\Lambda=0)$-vacuum equations on a bifurcate horizon in $(n+1)$-dimensions, $n\geq 3$.
Then the bifurcate horizon is a non-degenerate MKH of order 2 if and only if either (i) or (ii) hold:
\begin{enumerate}
\item[(i)] The metric takes the local form  $\gamma = dx^2 + x^p \delta$
where $\delta$ is the flat $n-2$ dimensional
  metric, and
  \begin{enumerate}
  \item[(a) ]  $ p = - \frac{2}{3}$ for $n=3$,
  \item[(b) ] $p=\frac{2}{n}\Big(1\pm\sqrt{\frac{2(n-1)}{n-2}}\Big)$ for $n\geq 4$.
  \end{enumerate}
  \item[(ii)] The metric takes the local form
  (\ref{metricsol}).
\end{enumerate}
\end{lemma}
\begin{remark}
  {\rm As any line element $\gamma= F(x) dx^2 + H(x) \delta$ has
    at least $\frac{1}{2}(n-2)(n-1)$ (local)  Killing vectors,
the emerging $(\Lambda=0)$-vacuum spacetimes in this Lemma admit, by Theorem~\ref{prop_KID_equations}, at least
$\frac{1}{2}(n-2)(n-1)+3$
(local) Killing vectors.}
\end{remark}

\begin{remark}
{\rm The spacetime in case (i) can be fully determined.
Recall the Kasner vacuum solution \cite{kasner} which can be written in the form
\begin{equation}
g=-x^{2p_0} (\mathrm{d}t)^2 +(\mathrm{d}x)^2+ \sum_{A=2}^nx^{2p_A} (\mathrm{d}y^A)^2
\,,\quad
 p_0+ \sum_{A=2}^n p_A=1\,, \quad (p_0)^2+  \sum_{A=2}^n (p_A)^2=1
\,.
\end{equation}
Now choose $p_A:= p/2$, $A=3,\dots n$, and $p_2=p_0=:q/2$, and set $u:=(t-y^2)/\sqrt{2}$, $v:=(t+y^2)/\sqrt{2}$.
Then
\begin{equation}
g=2x^{1-\frac{n-2}{2} p_{\pm}}\mathrm{d}u\mathrm{d}v+(\mathrm{d}x)^2+ x^{p_{\pm}}\delta
\quad \text{with} \quad    p_{\pm}=\frac{2}{n}\Big(1\pm\sqrt{2\frac{n-1}{n-2}}\Big)
\,.
\end{equation}
We conclude that the  data constructed in Lemma~\ref{kasner} (i) generate a subfamily of the Kasner metrics with  non-degenerate MKH of order 2
(for $n=3$  we have $p_+=2$ and this corresponds to the Minkowski metric which has a MKH of maximal order).
}
\end{remark}

\subsection{A no-go result}
\label{sec_non_deg_3+1}

In Section~\ref{sec_non_deg3} we have shown that whenever a $\Lambda$-vacuum spacetime admits a bifurcate Killing horizon with
a MKH  of order $m\geq 3$, the torsion one-form needs to be closed.
We work on the dense domain $\widetilde{S}$ where
the gauge condition \eq{special_gauge_master} can be imposed. The integrability condition \eq{integrability} becomes
\begin{equation}
{}^{\gamma}  R_{ABCD}D^D f  - D_{[A} f({}^{\gamma} R_{B]C}-\frac{2}{n-1}\Lambda \gamma_{B]C} )
=0
\,.
\label{int_cond_rew}
\end{equation}
%
Let us assume now that $m\geq n$, i.e.\ that the bifurcate horizon is a MKH of order $n$ or $n+1$.
Then \eq{master} admits at least $n-2$ independent non-constant solutions.
The set of points $q\in \widetilde{S}$
where all the gradients 
of these solutions are non-zero
and linearly independent is also a dense set. Locally near such points
we can construct an orthonormal frame $(e_{\widehat A})=(e_{\widehat 3}, e_{\widehat a})$, $\widehat a=4,\dots,n+1$, such the $e_{\widehat a}$'s are in the tangent plane generated by those gradients.
Then \eq{int_cond_rew} can be written in these frame components as
\begin{equation}
{}^{\gamma} R_{\widehat A\widehat B\widehat C\widehat a}  =\frac{1}{2}\delta_{\widehat A \widehat a}\left({}^{\gamma} R_{ \widehat B \widehat C}-\frac{2}{n-1}\Lambda \delta_{\widehat B\widehat  C} \right)
 - \frac{1}{2}\delta_{\widehat B \widehat a}\left({}^{\gamma} R_{ \widehat A \widehat C}-\frac{2}{n-1}\Lambda \delta_{ \widehat A \widehat C} \right)
\,.
\end{equation}
That yields
\begin{equation*}
{}^{\gamma} R_{\widehat A\widehat a}  =-\delta _{\widehat A \widehat a}\left({}^{\gamma} R-\frac{2(n-2)}{n-1}\Lambda \right)
\,,
\quad
{}^{\gamma}  R_{\widehat 3\widehat 3} =\frac{2(n-2)}{n(n-1)}\Lambda
\,,
\end{equation*}
whence
\begin{equation*}
{}^{\gamma} R=\frac{2(n-2)}{n}\Lambda \quad \Longrightarrow \quad {}^{\gamma} R_{AB}=\frac{2(n-2)}{n(n-1)}\Lambda \gamma_{AB}
\quad \Longrightarrow \quad {}^{\gamma} R_{ABCD} = \frac{4\Lambda}{n(n-1)} \gamma_{A[C}\gamma_{D]B}
\,,
\end{equation*}
i.e.\ the bifurcation surface $(S,\gamma)$ needs to be maximally symmetric.
Now we change the gauge.  Instead of  \eq{special_gauge_master} we transform to a gauge where the torsion one-form vanishes, which can be done at least locally.
Then we are in the setting of  Section~\ref{sec_vanish_torsion}, the emerging $\Lambda$-vacuum spacetime will have, at least locally, the maximum number of Killing vectors.


\begin{proposition}
Let $( M , g)$ be an $(n+1)$-dimensional  $\Lambda$-vacuum spacetime which  admits  a bifurcate Killing horizon such that one of its horizons
is a MKH of order $m\geq n$. Then $( M , g)$  is locally isometric to either Minkowski
or (Anti-)de Sitter, depending on the sign of the cosmological constant.
In particular a MKH of order $m=n$ which belongs to a bifurcate horizon does not exist.
\end{proposition}

Summarizing this section,
in $3+1$ dimensions we have  seen that  non-degenerate  MKHs of order 3 do not exist, at least not as part of a bifurcate Killing horizon,
while the maximal order 4 is obtained by maximally symmetric spacetimes.
The (Anti-)Nariai solution provides an example of a $\Lambda\lessgtr 0$-vacuum spacetime which admits a non-degenerate MKH of order 2.
An example of a  $(\Lambda=0)$-vacuum spacetime which admits a non-degenerate MKH of order 2 are certain Kasner metrics, cf.\  Section~\ref{sec_order2MKH}.

\section{Near-horizon geometries}
\label{sec_NHG}

Let us consider a $\Lambda$-vacuum spacetime $( M ,g)$ which admits a Killing vector field $\eta$ which has a degenerate Killing horizon $\mathcal{H}_1$.
In a neighborhood of $\mathcal{H}_1$ and in Gaussian null coordinates with $\mathcal{H}_1=\{v=0\}$ and $\eta=\partial_u$ the metric takes the form
\begin{equation}
g= 2\mathrm{d}u\mathrm{d}v  + 4 v \widehat s_A(v, x^B)\mathrm{d}u\mathrm{d}x^A + v^2h(v, x^A)\mathrm{d}u^2 + g_{AB}(v, x^C)\mathrm{d}x^A\mathrm{d}x^B
\,.
\label{gen_form}
\end{equation}
The associated \emph{near-horizon geometry} is
\begin{equation}
g_{\mathrm{NH}}= 2\mathrm{d}u\mathrm{d}v  + 4 v  \widehat s_A(0, x^B)\mathrm{d}u\mathrm{d}x^A + v^2h(0, x^A)\mathrm{d}u^2 + g_{AB}(0, x^C)\mathrm{d}x^A\mathrm{d}x^B
\,.
\label{NHmetric}
\end{equation}
Usually it is obtained as a certain limit of the original spacetime \eq{gen_form}, cf.\ \cite{KuLu}.
For our purposes, it is more enlightening how it can be obtained in terms of a characteristic initial value problem.
For this we consider the cut $S=\{u=0,v=0\}$ of the horizon $\mathcal{H}_1$, and we keep the same notation for the corresponding objects in the metric \eq{NHmetric}.
The null second fundamental forms ${}^{(1)}K$ induced by the line elements \eq{gen_form} or \eq{NHmetric} on
$\mathcal{H}^+_1:=\mathcal{H}_1\cap\{u>0\}$
vanish. The data induced on $S$
relevant for the characteristic initial-value problems are the $(n-1)$-dimensional Riemannian manifolds
$(S, \gamma_{AB} = g_{AB}|_S)$ and the torsion one-forms $\torsion_A= \widehat s_A|_S$. Hence, both metrics \eq{gen_form} and \eq{NHmetric} induce the same data on  $\mathcal{H}^+_1$ and $S$.
To generate either \eq{gen_form} or \eq{NHmetric} these data need to be supplemented with the 
null second fundamental forms induced by either \eq{gen_form} or \eq{NHmetric}
on the null hypersurfaces $\mathcal{H}^+_2:= \{ u=0\} \cap \{ v>0\}$ (again we keep the same notation for both metrics) generated by the null geodesics intersecting $\mathcal{H}_1$ transversally on $S$. These null second fundamental forms are in both cases proportional to $\partial_v g_{AB}$. Thus, such a second fundamental form is
in general different from zero for the case of the metric \eq{gen_form} but, in contrast, it vanishes for the case of the metric \eq{NHmetric}.
%
In other words, in view of a characteristic  initial-value problem,  taking the near-horizon limit  corresponds to replacing the induced null second fundamental form
on the initial hypersurface surface $N^+_2$ by trivial data. It is worth stating this as a
proposition.
\begin{proposition}
\label{dfn_NHG_alt}
Let $(M,g)$ be a $\Lambda$-vacuum spacetime which contains a degenerate Killing horizon $\mathcal{H}_1$, generated by a Killing vector $\eta$.
Then take any cut $S$ of  $\mathcal{H}_1$, and let $N_2$ denote the null hypersurface which is generated by the null geodesics through $S$
which are transversal to $\mathcal{H}_1$.
Denote by $({}^{(1)}K_{AB}\overset{\mathcal{H}_1}{=}0, {}^{(2)}K_{AB},g_{AB}|_S,
\bm{\torsion})$ the induced characteristic initial data. Then the associated NHG is the $\Lambda$-vacuum spacetime obtained as solution of the characteristic initial value problem for these data with  ${}^{(2)}K_{AB}$  replaced by trivial data.
\end{proposition}

It is clear from the construction (and Theorem \ref{theorem_fund_form}) that
the intersection surface $S$ becomes
the  bifurcation surface of a bifurcate Killing horizon. Thus
the construction automatically implies that the near horizon spacetime
admits an additional Killing vector field
$\xi$. In a sense, the present geometric reinterpretation of the
near horizon limit explains
why near horizon geometries always admit
an additional Killing vector $\xi = u \partial_ u - v \partial_v$, which vanishes at a bifurcation surface and generates a bifurcation horizon an
observation made in \cite{KuLu} (see also  \cite{LSW, PLJ,mps,mps2}).

In view of the KID equations for the bifurcate Killing horizon data,
the original Killing vector
$\eta$ corresponds to a solution of the KID equation
\eq{double_horizon_cond}
(or \eq{double_horizon_cond2}) which vanishes nowhere.
Thus, we can identify which characteristic data on a bifurcate horizon yield
near horizon geometries.

\begin{corollary}
  \label{LemmaNHG}
  Characteristic initial data on a bifurcate horizon $(\mathcal{H}_1\cup\mathcal{H}_2\cup S)$ consist of a Riemannian metric $\gamma$ and the torsion 1-form
  $\bm{\torsion}$ on $S$.
  The emerging $\Lambda$-vacuum spacetime is a NHG if and only if
  \eq{double_horizon_cond}
  or \eq{double_horizon_cond2} admits
a nowhere zero  solution.
\end{corollary}

\begin{remark}
  {\rm In the case
    $\Lambda=0$, this result was proved in \cite{LRS} by explicit construction of the
  spacetime. Indeed, given such data on $S$, there exists a global gauge where
  the torsion one-form satisfies  (\ref{special_gauge_master}). Then, one can define $h$ as in
    (\ref{dfn_h}),  write down the metric (\ref{NHmetric}) and check that
    it solves the $\Lambda$-vacuum field equations. This is a very direct way of establishing existence, but it is heavily limited by the condition
    that the solution of \eq{double_horizon_cond} is nowhere zero. Our existence results do not rely on this, so it seems reasonable to allow for fixed points of $\eta$.
  In fact, one of the main interests of viewing the near horizon
limit as a characteristic initial value problem is that it does not depend
on the fact that the Killing vector $\eta$ is non-zero on
$\mathcal{H}_1$, something which is crucial in the coordinate limit definition,
since it relies on Gaussian null coordinates adapted to $\eta$.
This suggests the existence of a natural generalization of the concept of NHG, which we put forward in Definition \ref{definition_NHG_alt} below.
}
\end{remark}


\begin{remark}
{\rm
  Another interesting characterization of
  vacuum near horizon geometries alternative
  to Corollary \ref{LemmaNHG} is given in
\cite{LSW}, where it is shown that a vacuum spacetimes is a near
horizon geometry if and only if
admits a foliation of non-expanding  horizons (supplemented by a transverse
one).}
\end{remark}

It follows from the KID equations and the fact
that the  master equation \eq{Eq} must hold on $S$  for each  Killing vector field $\eta^{(k)}$ of the original spacetime
for which  $\mathcal{H}_1$ is a degenerate horizon that the associated solutions $f^{(k)}$ of \eq{double_horizon_cond}
(which need to be independent) generate, by Corollary~\ref{cor_multiple_horizons},  independent  Killing vectors of the near-horizon geometry too, for which  $\mathcal{H}^+_1$ is also a degenerate horizon.
We recover \cite[Theorem~4.1]{mps} in the $\Lambda$-vacuum case:

\begin{corollary}
Consider a $\Lambda$-vacuum spacetime which admits a fully degenerate (multiple) Killing horizon of order $m\geq 1$. Then there exists an
associated $\Lambda$-vacuum spacetime, the near-horizon limit, which has a non-degenerate MKH  of order $m+1$.
\end{corollary}

\subsection{Near-horizon limit if $\dim( \Kill_\H^{deg})\geq 2$}

Let us now assume that for  the metric \eq{gen_form} $\mathcal{H}_1=\{v=0\}$ is a MKH and that there exists another Killing vector $\widetilde \eta$ for which
$\H_{\widetilde \eta}$
is a  degenerate  Killing horizon satisfying $\overline{\H}_{\widetilde \eta} = \overline{\H_1}$, i.e.\ we assume  $\dim( \Kill_{\H_1}^{deg})\geq 2$ (recall its definition
in Theorem~\ref{thm_MKH2}).
Then the near-horizon limit can be taken with respect to both $\eta$ or $\widetilde \eta$.
In the latter case one would first need to transform to Gaussian null coordinates $(\widetilde u, \widetilde v, \widetilde x^A)$. However, here we only need to observe the following: with $\widetilde \eta\overset{\mathcal{H}_1}{=}f\eta$, we find that, at points where $f$
is not zero,
%
\begin{equation}
\widetilde g_{AB}\overset{\mathcal{H}_1} {=} g_{AB}\,,
\quad\bm{ \widetilde {\widehat s}}\overset{\mathcal{H}_1} {=} \bm{ \widehat  s}-\mathrm{d}\log |f|
\,,
\end{equation}
and this implies that the characteristic Cauchy data for the near-horizon geometries associated to   $\eta$ and $\widetilde \eta$ coincide up to gauge transformations (cf.\ Section~\ref{sec_bifurcate}).
This way we  recover \cite[Theorem~7]{mps2} in the $\Lambda$-vacuum case:
\begin{corollary}
Let (M,g) be a $\Lambda$-vacuum spacetime containing a MKH  $\mathcal{H}$ and let  $\eta,\widetilde \eta\in \Kill_\H^{deg}$.
Then the near-horizon geometries w.r.t.\ $\eta$ and $\widetilde \eta$ are locally isometric.
\end{corollary}

\subsection{Generalized Near Horizon Geometry}

Consider a spacetime admitting a null hypersurface
${\mathcal N}$ and a (non-trivial) Killing vector $\eta$ which, at ${\mathcal N}$, is tangent and null. Since $\eta$ vanishes at most on codimension-two submanifolds,
there is a Kiling horizon $\mathcal{H}_1$ of $\eta$
satisfying $\overline{\mathcal{H}_1} = {\mathcal N}$.
Assume further that this horizon is degenerate.
%
To determine the NHG in the usual way we need to restrict to (connected components of) $\mathcal{H}_1$ since fixed points of $\eta$ have to be removed
in order to construct  Gaussian null coordinates. However,
our construction in terms of characteristic initial data  provides an immediate alternative definition of NHG for $\Lambda$-vacuum,
which does not care about fixed points of $\eta$:

\begin{definition}
\label{definition_NHG_alt}
Let $(M,g)$ be $\Lambda$-vacuum and
admit a connected null hypersurface ${\mathcal N}$ and a Killing vector $\eta$ as described above.
Consider any cut $S$ of ${\mathcal N}$ and let $\gamma$ be the induced metric of $S$
and $\bm{\torsion}$ the torsion one-form with respect to any choice
of null frame.
A {\bf generalized near horizon limit}
spacetime of $(M,g)$ is a spacetime
obtained by solving the characteristic initial value corresponding to a bifurcate
Killing horizon and data $(S,\gamma,\bm{\torsion})$.
\end{definition}

\begin{remark}
{\rm It may appear from the wording of the Definition that there are many possible generalized near horizon limits, but this
is merely a consequence of the fact that, to the best of our knowledge, there is no existence result of
a unique maximal Cauchy development in the characteristic case. All spacetimes generated by the data
are isometric in some neighbourhood of $S$.}
\end{remark}

\begin{remark}
{\rm   When $\eta$ has no fixed points, the generalized near horizon geometry is the same as the standard
  near horizon limit.}
\end{remark}

\begin{remark}
{\rm   By Corollary \ref{cor_multiple_horizons}, the null hypersurface ${\mathcal N}$ of the generalized near horizon limit is a non-fully degenerate multiple Killing horizon of order at least two.}
\end{remark}


\section{Fully degenerate MKHs}
\label{sec_fully_deg}
So far we have studied the emergence of $\Lambda$-vacuum spacetimes with MKHs  in terms of a characteristic
initial value problem  with data given on a bifurcate Killing horizon. However, since the bifurcate horizon is necessarily non-degenerate with respect to the bifurcate Killing vector this approach permits
merely the construction of non-degenerate MKHs.
In this section we consider the case of fully degenerate MKHs, and obtain partial results.

In order to construct spacetimes with fully degenerate MKHs  we need to modify the initial data in such a way that a
bifurcate Killing vector does not arise. For that purpose, one of the initial null hypersurfaces, $\mathcal{H}_2^+$ say, needs to have a non-vanishing null second fundamental form.
Since this null hypersurface will therefore not be a Killing horizon of the emerging spacetime anymore, we will  denote it by $N_2^+$.
It turns out, though,  that an analysis of the KID equations becomes rather intricate for such a general class of data: while the KID equations \eq{KID_gen1}, \eq{KID_gen2} and  \eq{KID_gen4} still provide ODEs for the Killing candidate on $\mathcal{H}^+_1$ and $N_2^+$, the main problem is to arrange the data in such a way that \eq{KID_gen3} holds. In the case of a bifurcate horizon we have seen, cf.\ the paragraph after Remark \ref{rem}, that its second-order $u$- (respectively, $v$-) derivative vanishes on $\mathcal{H}^+_1$ (resp.\ $\mathcal{H}^+_2$) whence \eq{KID_gen3}  becomes a constraint which only provides restrictions on $S$. This is no longer true for fully degenerate MKHs.

%
%
%

Fortunately, an analysis of the equations shows that their behavior is similar if additional conditions are imposed on the initial data. Hence, the strategy we are going to apply is as follows: first, in subsection \ref{subsec:Knonzero} we basically restrict attention to data where the second fundamental form is independent of the $x^A$-coordinates (in an appropriate gauge), i.e.\ it is constant along appropriately selected cuts. We will further impose that these cuts are intrinsically flat and set $\Lambda =0$ and a vanishing torsion one-form. In this way, a class of data generating fully degenrate MKHs is identified, and these data include solutions corresponding to gravitational plane waves. In subsection~\ref{subsec_vanishing_shear} we will restore a general $\Lambda$ and explore the Ansatz where the shear tensor vanishes on $N_2^+$, with no conclusive results concerning fully degenerate MKHs..

\subsection{KID equations for non-vanishing null second fundamental form}\label{subsec:Knonzero}

The following result is proved in Appendix~\ref{app_proof1}:

\begin{proposition}
\label{prop_KID_equations2}
Consider two smooth hypersurfaces, $\mathcal{H}^+_1$ and $N^+_2$ in an $(n+1)$-dimensional manifold with transverse intersection along a smooth $(n-1)$-dimensional submanifold $S$ in adapted null coordinates and where the gauge conditions \eq{cond_gauge}-\eq{cond_gauge2} are fulfilled.
Consider characteristic initial data which satisfy
\begin{equation*}
{}^{(1)}\pi_{AB} \overset{\mathcal{H}^+_1}{=} 0
\,, \quad
{}^{(1)}\theta \overset{S}{=}0
\,, \quad
\partial_C{}^{(2)}\pi_{AB}\overset{N_2^+}{=}0
\,, \quad
{}^{(2)}\theta \overset{S}{ =:}\theta_S= \mathrm{const.}
\,,
\end{equation*}
as well as
\begin{equation*}
\Lambda=0\,, \quad  (S,\gamma ) \cong \text{Euclidean space}\,, \quad {\bf \varsigma}=0
\,.
\end{equation*}
Then the emerging vacuum spacetime admits  $k$ Killing vectors if and only if
the following set of KID equations,
\begin{eqnarray*}
\partial_v\zeta_v  &\overset{N_2^+}{=}&0
\,,
\\
\Big(\partial_{v} - \frac{2}{n-1}{}^{(2)}\theta\Big)\zeta_A
-2{}^{(2)}\pi_{A}{}^B\zeta_B&\overset{N_2^+}{=}&  -\partial_{A} \zeta_{v}
\,,
\\
\partial_v\partial_v \big(\zeta^v -\frac{1}{2}g^{vv}\zeta_v\big) &\overset{N_2^+}{=}&0
\,,
\\
(\widehat\nabla_{(A}\zeta_{B)} )_{\mathrm{tf}}
&\overset{N_2^+}{=}&
-{}^{(2)}\pi_{AB}\big(\zeta^v -\frac{1}{2}g^{vv}\zeta_v\big)
\,,
\\
\partial_{(u}\zeta_{v)} &\overset{S}{=}& \Gamma^{u}_{uv}\zeta_{u}+ \Gamma^{v}_{uv}\zeta_{v}
\,,
\\
D^A\zeta_A &\overset{S}{=}& -\theta_S\zeta_{u}
\,,
\\
 \Delta_{\gamma}\zeta_{v} &\overset{S}{=}&
-\theta_S\partial_{[u}\zeta_{v]}
-\frac{\theta^2_S}{n-1}\zeta_{u}
-|{}^{(2)}\pi|^2 \zeta_{u}
\,,
\\
\partial_A\partial_{[u}\zeta_{v]}   &\overset{S}{=}& -\frac{\theta_S}{n-1} \partial_{A}\zeta_{u}
-{}^{(2)}\pi_{A}{}^B\partial_B\zeta_u
\\
D_{A}D_B\zeta_{u} &\overset{S}{=}& 0
\,,
\end{eqnarray*}
 admits $k$ independent solutions, determined by  data  $\zeta_{\mu}|_S$ and an additive integration constant
which arises from  the $ \partial_{[u}\zeta_{v]}|_S$-equation. The data on $\mathcal{H}^+_1$ are then uniquely determined by  \eq{KID_deg1}-\eq{KID_deg3} in Appendix~\ref{app_proof1} below.
\end{proposition}


Now, we use the previous proposition to try and derive initial data producing MKHs (not necessarily degenerate). To that end, we add initial data for $\zeta$:
 Recalling
(\ref{proportionality}) and (\ref{horizon_conds}) we keep $\zeta_v$ as the only non-zero component on $S$ (as this is necessary for the Killing vector field to have $\mathcal{H}_1$ as a Killing horizon) so that as a corollary of Proposition~\ref{prop_KID_equations2} we obtain:
\begin{corollary}
\label{cor_fully_deg}
Under the same hypotheses as in Proposition~\ref{prop_KID_equations2}, the emerging vacuum spacetime
has $\mathcal{H}^+_1$ as  a MKH of order $m$ if and only if the following set of equations admits $m$ independent solutions,
parameterized by $(\eta_u|_S=\eta_A|_S=0,\eta_v|_S ,\partial_{[u}\eta_{v]}|_S= \kappa_{\eta} =\mathrm{const.})$,
\begin{eqnarray}
\partial_v\eta_v  &\overset{N_2^+}{=}&0
\,,
\label{KID_towards_deg1}
\\
\Big(\partial_{v} - \frac{2}{n-1}{}^{(2)}\theta\Big)\eta_A
 &\overset{N_2^+}{=}&  2{}^{(2)}\pi_{A}{}^B\eta_B-\partial_{A} \eta_{v}
\,, \quad \eta_A\overset{S}{=}0
\,,
\label{KID_towards_deg2}
\\
\partial_v\partial_v \big(\eta^v -\frac{1}{2}g^{vv}\eta_v\big)  &\overset{N_2^+}{=}&0
\,,\quad
\big(\eta^v -\frac{1}{2}g^{vv}\eta_v\big)\overset{S}{=}0\,, \quad \partial_v \big(\eta^v -\frac{1}{2}g^{vv}\eta_v\big)\overset{S}{=} - \kappa_{\eta}
\,,
\label{KID_towards_deg3}
\\
(\widehat\nabla_{(A}\eta_{B)} )_{\mathrm{tf}}
 &\overset{N_2^+}{=}&
-{}^{(2)}\pi_{AB}\big(\eta^v -\frac{1}{2}g^{vv}\eta_v\big)
\,,
\label{KID_towards_deg4}
\\
 \Delta_{\gamma} \eta_{v}&\overset{S}{=}&
-\theta_S \kappa_{\eta}
\label{KID_towards_deg6}
\,.
\end{eqnarray}
%
%
%
\end{corollary}

\begin{remark}
{\rm
The condition on $\partial_{(u}\eta_{v)} |_S$ in Proposition~\ref{prop_KID_equations2}  is  relevant to determine
the initial data on $S$ for the second-order ODEs for $\eta^u|_{N^+_1}$ and $\eta^v|_{N^+_2}$ and is thus
contained in the initial conditions in \eq{KID_towards_deg3}.
}
\end{remark}

As in previous cases, $\kappa_\eta$ represents the surface gravity relative to $\eta$ of the emerging horizon. Hence, in order to get a fully degenerate MKH we need to make sure that every solution of this  system satisfies $\kappa_{\eta}=0$. To that end, it is convenient to assume data which satisfy
${}^{(2)}\pi_{AB}|_S=0$.
Then, we determine the first- and second-order $v$-derivative of   \eq{KID_towards_deg4} which yields with \eq{KID_towards_deg2}-\eq{KID_towards_deg3} and evaluated on $S$
\begin{eqnarray*}
(\partial_{A}\partial_{B}\eta_{v})_{\mathrm{tf}}|_S
 &=&
0
\,,
\\
\frac{1}{n-1}\theta_S (\partial_{A}\partial_{B}\eta_{v})_{\mathrm{tf}}|_S
 &=&
-\kappa_{\eta}\partial_v{}^{(2)}\pi_{AB}
\,,
\end{eqnarray*}
which implies
\begin{equation*}
\kappa_{\eta}\partial_v{}^{(2)}\pi_{AB}|_S=0
\,.
\end{equation*}
Clearly, $\kappa_{\eta}=0$ must hold whenever $\partial_v{}^{(2)}\pi_{AB}|_S\ne 0$.
%
%
To enforce $\kappa_{\eta}=0$ let us therefore restrict attention to initial data of the form
\begin{equation}
{}^{(2)}\pi_{AB}|_S=0\,, \quad \partial_v{}^{(2)}\pi_{AB}|_S\ne 0
\,.
\label{restrictions_data}
\end{equation}
By the latter condition we mean that $\partial_v{}^{(2)}\pi_{AB}|_S$ is not identically zero, i.e. that at least one component (which is, by assumption, constant on $S$) is non-zero.

We conclude from Corollary~\ref{cor_fully_deg} that  the emerging vacuum spacetime
has $\mathcal{H}^+_1$ as  a fully degenerate  Killing horizon of order $m$ if
the following set of equations admits $m$ independent solutions
$\eta_v=\overset{(0)}{\eta}(x^A)$  (with $\eta_A|_S =0$)
\begin{eqnarray}
\Big(\partial_{v} - \frac{2}{n-1}{}^{(2)}\theta\Big)\eta_A
-2{}^{(2)}\pi_{A}{}^B\eta_B &\overset{N_2^+}{=}&  -\partial_{A}\overset{(0)}{\eta}
\,,
\label{KID_towards_deg2B}
\\
(\widehat\nabla_{(A}\eta_{B)} )_{\mathrm{tf}}
&\overset{N_2^+}{=}&
0
\,,
\label{KID_towards_deg4B}
\\
 \Delta_{\gamma}\overset{(0)}{\eta}&\overset{S}{=}&
0
\,.
\label{KID_towards_deg5B}
\end{eqnarray}
Using \eq{KID_towards_deg4B}, the anti-symmetrized derivative  of \eq{KID_towards_deg2B}  reads
\begin{equation*}
\Big(\partial_{v} - \frac{2}{n-1}{}^{(2)}\theta\Big)\partial_{[A}\eta_{B]}
+{}^{(2)}\pi_{A}{}^C\partial_{[B}\eta_{C]}
- {}^{(2)}\pi_{B}{}^C\partial_{[A}\eta_{C]} \overset{N_2^+}{=} 0
\,.
\end{equation*}
As the initial data on $S$ for this equation vanish (cf. (\ref{KID_towards_deg2})) we deduce
\begin{equation*}
\partial_{[A}\eta_{B]}  \overset{N_2^+}{=} 0
\,.
\end{equation*}
Using \eq{KID_towards_deg4B} we also  compute the symmetrized trace-free part and the divergence of \eq{KID_towards_deg2B},
\begin{align}
(\widehat\nabla_A\widehat\nabla_{B}\overset{(0)}{\eta})_{\mathrm{tf}}
  \overset{N_2^+}{=}   0
\,,
\label{itermed_deg}
\qquad
\Big(\partial_{v} - \frac{2}{n-1}{}^{(2)}\theta\Big)\widehat\nabla^A\eta_A
 \overset{N_2^+}{=}   -\widehat\Delta\overset{(0)}{\eta}
\,.
\end{align}
The first relation implies (note that the Christoffel symbols of $g_{AB}|_{N_2^+}$ vanish),
%
\begin{align*}
&\partial_v\widehat\Delta\overset{(0)}{\eta}  \overset{N_2^+}{=}  \partial_vg^{AB}\partial_A\partial_B\overset{(0)}{\eta}
 \overset{N_2^+}{=}  - \frac{2}{n-1}{}^{(2)}\theta\widehat\Delta \overset{(0)}{\eta}
\\
&\overset{ \eq{KID_towards_deg5B}}{\Longrightarrow} \quad\widehat\Delta\overset{(0)}{\eta}   \overset{N_2^+}{=} 0
\quad \Longrightarrow\quad
\widehat\nabla^A\eta_A  \overset{N_2^+}{=} 0
\,.
\end{align*}
%
Any solution of \eq{KID_towards_deg2B}-\eq{KID_towards_deg5B}  is
therefore necessarily of the form

\begin{equation*}
\partial_{A}\partial_{B}\overset{(0)}{\eta}
 \overset{N_2^+}{=}   0
\,,
\qquad
\partial_{A}\eta_{B}
 \overset{N_2^+}{=}
0
\,.
\end{equation*}

The ODE \eq{KID_towards_deg2B} can be integrated for data $\eta_A\overset{S}{=}0$. As the source and the coefficients do not depend on the $x^A$'s the same will be true for $\eta_A|_{N_2^+}$, so that there are no further obstructions by \eq{KID_towards_deg4B}.
We deduce that the above system has $m=n$ independent solutions which are given by
\begin{equation*}
\overset{(0)}{\eta} \overset{N_2^+}{=} 1 \quad \text{and} \quad \overset{(0)}{\eta} \overset{N_2^+}{=} x^A\,, \enspace A=3,\dots,n+1
\,,
\end{equation*}
together with the corresponding solution $\eta_A$ of
\eq{KID_towards_deg2B} 
with $\eta_A|_S=0$.

We thus have:
\begin{proposition}
\label{prop_KID_equations3}
Consider two smooth hypersurfaces $\mathcal{H}^+_1=\{v=0\}$ and $N_2^+=\{u=0\}$  in an $(n+1)$-dimensional manifold with transverse intersection along a smooth $(n-1)$-dimensional submanifold $S$ in adapted null coordinates
and where the gauge conditions \eq{cond_gauge}-\eq{cond_gauge2} are fulfilled.
Consider characteristic initial data which satisfy
\begin{align*}
{}^{(1)}\pi_{AB} \overset{\mathcal{H}^+_1}{=}0
\,, \quad
{}^{(1)}\theta \overset{S}{=}0
\,,
\quad
\partial_C{}^{(2)}\pi_{AB}\overset{N_2^+}{=}0
\,,
 \quad {}^{(2)}\pi_{AB}\overset{S}{=}0\,, \quad \partial_v{}^{(2)}\pi_{AB}|_S\ne 0
\,, \quad{}^{(2)}\theta\overset{S}{=}\mathrm{const.}
\end{align*}
and where
\begin{equation*}
\Lambda=0\,, \quad (S,\gamma) \cong \text{Euclidean space}\,, \quad  {\bf \varsigma}=0
\,.
\end{equation*}
Then the emerging vacuum spacetime
has $\mathcal{H}^+_1$ as  a fully degenerate  multiple Killing horizon of  maximal order $n$.
\end{proposition}


In Appendix~\ref{app_number_killings} we prove the following proposition which addresses the issue how  many independent Killing vectors the spacetimes generated in  Proposition~\ref{prop_KID_equations3} have.

\begin{proposition}
\label{prop_KID_equations4}
%
Under the same hypotheses as in Proposition~\ref{prop_KID_equations3} we have:
\begin{enumerate}
\item[(i)]
The emerging vacuum spacetime has
at least   $2n-1$ Killing vectors.
\item[(ii)]
For appropriate data there may be a maximum of $\frac{1}{2}(n-1)(n-2)+1$
 additional Killing vectors, depending on whether the candidate fields $\overset{(1)}\zeta_{AB}(v)$ determined by \eq{eqn_coeff} and \eq{ev_anti_sym}
fulfill \eq{ev_sym}, cf.\   Appendix~\ref{app_number_killings}.
\end{enumerate}
\end{proposition}

Given the properties of gravitational plane waves concerning MKHs (\cite[subsection~4.4]{mps}) we deduce that such wave solutions are generated by initial data of the type found in Proposition~\ref{prop_KID_equations3}.

\subsection{Vanishing shear}
\label{subsec_vanishing_shear}

 Again we aim  at the construction of MKHs
where the initial data do \emph{not} lead to a bifurcate horizon.
Before doing that, though, let us simplify the KID equations \eq{KID_gen1}-\eq{KID_gen9}
for  a Killing vector for which $\mathcal{H}^+_1$ is a Killing horizon
 in a more general setting.
The  Killing vector needs to satisfy $\zeta_u|_S=\zeta_A|_S=0$ and one checks that
 the system \eq{KID_gen1}-\eq{KID_gen9} becomes (we assume that the gauge conditions
\eq{cond_gauge}-\eq{cond_gauge2} are fulfilled),
%
\begin{eqnarray}
\zeta_u  &\overset{\mathcal{H}^+_1}{=}& 0
\,, \quad
\zeta_{A}\,\overset{\mathcal{H}^+_1}{=}\,0
\,,
\label{KID_deg_hor1}
\\
\text{2nd-order ODE for $\zeta^u|_{\mathcal{H}^+_1}$}\,,&&  \zeta^u|_S=\overset{(0)}\zeta(x^A)\ne 0
\,, \quad
\partial_u\zeta^u \overset{S}{=} \partial_{[u}\zeta_{v]}
\,,
\label{2nd-order1}
\\
\zeta_v  &\overset{N_2^+}{=}& \overset{(0)}\zeta
\,,
\\
\Big(\partial_{v}- \frac{2}{n-1}{}^{(2)}\theta\Big) \zeta_{A}&\overset{N_2^+}{=}&
  2\, {}^{(2)}\pi_{A}{}^B\zeta_{B}
- (\partial_{A}- 2\Gamma^{v}_{vA} ) \overset{(0)}\zeta
\,,
\\
(\widehat\nabla_{(A}\zeta_{B)})_{\mathrm{tf}} &\overset{N_2^+}{=}& \frac{1}{2}\Big(({}^{(2)}\Xi_{AB})_{\mathrm{tf}}+g^{vv} {}^{(2)}\pi_{AB}\Big)\overset{(0)}\zeta
-{}^{(2)}\pi_{AB}\zeta^v
\,,
\\
\text{2nd-order ODE for $\zeta^v|_{N_2^+}$}\,,&& \zeta^v|_S = 0
\,, \quad
 \partial_v\zeta^v \overset{S}{=} -\partial_{[u}\zeta_{v]}-  \partial_ug_{uv}\overset{(0)}\zeta
\, ,
\label{2nd-order2}
\\
\partial_{(u}\zeta_{v)} &\overset{S}{=}&\partial_ug_{uv}\overset{(0)}\zeta
\,,
\label{KID_deg_hor8}
\\
D^A\partial_v\zeta_A
&\overset{S}{=}& \theta_S\partial_{[u}\zeta_{v]}
+\frac{1}{2}\partial_v\mathrm{tr}({}^{(2)}\Xi)\overset{(0)}\zeta
\,,
\nonumber
\\
D_A\partial_{[u}\zeta_{v]}
                      &\overset{S}{=}& 0
\,.
\nonumber
\end{eqnarray}
Using \eq{constr3} we rewrite the last two equations as
\begin{eqnarray}
\Delta_{\gamma}\overset{(0)}\zeta+ \kappa\theta_S
&\overset{S}{=}&
 2\varsigma_AD^A\overset{(0)}\zeta+\overset{(0)}\zeta(D^A-\varsigma^A)\varsigma_A
+\frac{1}{2}\Big(}{^{\gamma}  R - 2\Lambda\Big)\overset{(0)}\zeta
\,,
\\
\partial_{[u}\zeta_{v]}
                      &\overset{S}{=}& \kappa=\mathrm{const.}
\label{KID_deg_hor10}
\end{eqnarray}

Here we choose to analyze the situation where $N_2^+$ has vanishing shear while its expansion
$\theta_S$ on the intersection manifold $S$ as well as $(S,\gamma)$ itself can be arbitrary, i.e. we consider data of the form
\begin{equation}
{}^{(1)}\pi_{AB} \overset{\mathcal{H}^+_1}{=}0
\,, \quad
{}^{(1)}\theta \overset{S}{=}0
\,, \quad
{}^{(2)}\pi_{AB} \overset{N^+_2}{=}0
\,,\quad
{}^{(2)}\theta \overset{S}{=}\theta_S\ne 0
\,.
\label{ass_vanishing_shear1}
\end{equation}
Observe that the 2nd-order ODEs \eq{2nd-order1} and \eq{2nd-order2}
 admit unique solutions  for $\zeta^u|_{\mathcal{H}^+_1}$ and  $\zeta^v|_{N^+_2}$ whose precise form is irrelevant whenever the shear tensors of both
null hypersurfaces vanish, as they only arise in the other equations accompanied by a factor of ${}^{(a)}\pi_{AB}$.
Thus, we do not need to consider the 2nd-order ODEs \eq{2nd-order1} and \eq{2nd-order2}  any further.
Let us further assume that the torsion one form is given by
\begin{equation}
{\bf \varsigma}=\mathrm{d}\log\theta_S
\,.
\label{ass_vanishing_shear2}
\end{equation}
This choice makes the solution
of \eq{constr2} as given in Appendix~\ref{app_kids_deg_van_shear} as simple as possible.
In that Appendix~\ref{app_kids_deg_van_shear} we prove:
\begin{proposition}
\label{prop_KID_equations5}
Consider two smooth hypersurfaces $\mathcal{H}^+_1=\{v=0\}$ and $N_2^+=\{u=0\}$  in an $(n+1)$-dimensional manifold with transverse intersection along a smooth $(n-1)$-dimensional submanifold $S$ in adapted null coordinates and where the gauge conditions \eq{cond_gauge}-\eq{cond_gauge2} are fulfilled.
Consider characteristic initial data
 which satisfy
\begin{equation}
{}^{(1)}\pi_{AB} \overset{\mathcal{H}^+_1}{=}0
\,, \quad
{}^{(1)}\theta \overset{S}{=}0
\,, \quad
{}^{(2)}\pi_{AB} \overset{N^+_2}{=}0
\,,\quad
{}^{(2)}\theta\overset{S}{=}\theta_S\ne 0
\,, \quad
{\bf \varsigma}=\mathrm{d}\log\theta_S
\,.
\end{equation}
Then $\mathcal{H}^+_1=\{v=0\}$ is a multiple Killing horizon of order $m$ if and only if
\begin{align}
({}^{\gamma} R_{AB})_{\mathrm{tf}} =0\label{einstein}
\,,
\end{align}
and the following system admits $m$ independent solutions $(f, \kappa)$  on $S$ with $\kappa=\mathrm{const.}$,
%
\begin{align}
(D_{A}D_{B}f)_{\mathrm{tf}}&=0
\,,
\\
(  D_{(A}\theta_SD_{B)}f )_{\mathrm{tf}} &=0
\\
\Delta_{\gamma} f
+ \kappa
&= \frac{1}{2 }\Big({}^{\gamma} R - 2\Lambda\Big)f
\,.
\end{align}
The  multiple Killing horizon is fully degenerate if and only if all solutions satisfy  $\kappa=0$.
\end{proposition}


\begin{remark}
{\rm
At this stage, it is not clear to us whether or not there exist data for which Proposition  \ref{prop_KID_equations5} produces fully degenerate MKHs.
}
\end{remark}

\section*{Acknowledgments}
MM acknowledges financial support under projects
FIS2015-65140-P (Spanish MINECO/FEDER) and
SA083P17 (Junta de Castilla y Le\'on).
TTP acknowledges financial support by the Austrian Science Fund (FWF)
P~28495-N27.
JMMS is supported under Grants No. FIS2017-85076-P (Spanish MINECO/AEI/FEDER, UE) and No. IT956-16 (Basque Government).

\appendix

\section{Some calculations relevant for Section~\ref{sec_fully_deg}}

\subsection{Proof of Proposition~\ref{prop_KID_equations2}}
\label{app_proof1}

In order to establish Proposition~\ref{prop_KID_equations2}
let us start with data of the form
\begin{equation}
{}^{(1)}\pi_{AB} \overset{\mathcal{H}^+_1}{=}0
\,, \quad
{}^{(1)}\theta \overset{S}{=}0
\,, \quad
\partial_A|{}^{(2)}\pi|^2\overset{N^+_2}{=}0
\,, \quad
\widehat\nabla_B{}^{(2)}\pi_A{}^B \overset{N^+_2}{=}0
\,, \quad
{}^{(2)}\theta\overset{S}{=:}\theta_S= \mathrm{const.}
\,,
\label{restriction_data1}
\end{equation}
and let us further assume that
\begin{equation}
\widehat  R_{AB} =\frac{2}{n-1}\Lambda g_{AB} \,, \quad {\bf \varsigma}=0
\,.
\label{restriction_data2}
\end{equation}
As the data will be more restricted later on, we do not care so much how data which fulfill \eq{restriction_data1} are constructed
from a family of Riemannian metrics as in Theorem~\ref{thm_Cauchy_problem}.

We analyze the KID equations \eq{KID_gen1}-\eq{KID_gen9} for this class of data, where we again assume a gauge where
\eq{cond_gauge}-\eq{cond_gauge2} holds,
\begin{equation}
\Gamma^u_{uu} \overset{\mathcal{H}^+_1} {=} 0\,, \quad
\Gamma^v_{vv}\overset{N_2^+}{ =} 0\,, \quad
g_{12}\overset{S}{=}1
\,.
\label{cond_gauge_again}
\end{equation}
First of all we observe that the constraint equations \eq{constr1}-\eq{constr4} yield,
\begin{align}
{}^{(1)}\theta\overset{\mathcal{H}^+_1}{=}& 0
\,,
\quad
\Gamma^u_{uA}  \overset{\mathcal{H}^+_1}{=}0
\,,
\\
\mathrm{tr}({}^{(1)}\Xi)   \overset{\mathcal{H}^+_1}{=}& -2\theta_S
\,,
\quad
({}^{(1)}\Xi_{AB})_{\mathrm{tr}} \overset{\mathcal{H}^+_1}{=} -2{}^{(2)}\pi_{AB}|_S
\,,
\\
\Big(\partial_v+\frac{{}^{(2)}\theta}{n-1}\Big) {}^{(2)}\theta \overset{N_2^+}{=}& -|{}^{(2)}\pi|^2
\,, \quad {}^{(2)}\theta \overset{S}{=} \theta_S
\\
\Gamma^v_{vA} \overset{N_2^+}{=}&0
\,,
\quad
{}^{(2)}\Xi_{AB}
\overset{N_2^+}{=}0
\label{constr_tr_Xi}
\,.
\end{align}
Moreover, a  computation reveals that
\begin{align*}
R_{Auu}{}^u\overset{\mathcal{H}^+_1}{=}& 0
\,,
\\
R_{vuu}{}^u\overset{\mathcal{H}^+_1}{=}& R_{uv} -R_{v Au}{}^A  \overset{\mathcal{H}^+_1}{=}
-\frac{1}{2}g_{uv}\Big(\partial_u\mathrm{tr}({}^{(1)}\Xi)+{}^{(1)}K^{AB}{}^{(1)}\Xi_{AB}-\frac{4}{n-1}\Lambda \Big)
\overset{\mathcal{H}^+_1}{=}
\frac{2 \Lambda}{n-1} g_{uv}
\,,
\\
R_{Avv}{}^v\overset{N_2^+}{=}& 0
\,,
\\
R_{uvv}{}^v\overset{N_2^+}{=}& R_{uv} -R_{u Av}{}^A
\overset{N_2^+}{=}
\frac{2}{n-1}\Lambda g_{uv}
\,.
\end{align*}
%
%
%
%

The KID equations \eq{KID_gen1}-\eq{KID_gen4} on $\mathcal{H}^+_1$ become
\begin{align}
\partial_u\zeta_u \overset{\mathcal{H}^+_1}{=}&0
\,,
\label{KID_deg1}
\\
\partial_{u} \zeta_{A}\overset{\mathcal{H}^+_1}{=}& -\partial_{A} \zeta_{u}
\,,
\label{KID_deg2}
\\
\partial_u\partial_u \big(\zeta^u -\frac{1}{2}g^{uu}\zeta_u\big)\overset{\mathcal{H}^+_1}{=}& \frac{2}{n-1}\Lambda \zeta_u
\,,
\label{KID_deg3}
\\
(\widehat\nabla_{(A}\zeta_{B)} )_{\mathrm{tf}}\overset{\mathcal{H}^+_1}{=}& -{}^{(2)}\pi_{AB}|_S\zeta_{u}
\,.
\label{KID_deg4}
\end{align}
On $N_2^+$ they read
\begin{align}
\partial_v\zeta_v \overset{N_2^+}{=}&0
\,,
\label{KID_deg5}
\\
\Big(\partial_{v} - \frac{2}{n-1}{}^{(2)}\theta\Big)\zeta_A
-2{}^{(2)}\pi_{A}{}^B\zeta_B\overset{N_2^+}{=}&  -\partial_{A} \zeta_{v}
\,,
\label{KID_deg6}
\\
\partial_v\partial_v \big(\zeta^v -\frac{1}{2}g^{vv}\zeta_v\big) \overset{N_2^+}{=}&\frac{2}{n-1}\Lambda \zeta_v
\,,
\label{KID_deg7}
\\
(\widehat\nabla_{(A}\zeta_{B)} )_{\mathrm{tf}}
\overset{N_2^+}{=}&
-{}^{(2)}\pi_{AB}\big(\zeta^v -\frac{1}{2}g^{vv}\zeta_v\big)
\,.
\label{KID_deg8}
\end{align}
For the KID equations  \eq{KID_gen5}-\eq{KID_gen9} on $S$ a computation shows
(note that $ R_{vuA}{}^{\mu} \overset{S}{=}0$)
\begin{align}
\partial_{(u}\zeta_{v)}\overset{S}{=}& \Gamma^{u}_{uv}\zeta_{u}+ \Gamma^{v}_{uv}\zeta_{v}
\,,
\label{KID_deg9}
\\
D^A\zeta_A\overset{S}{=}& -\theta_S\zeta_{u}
\,,
\label{KID_deg10}
\\
\Delta_{\gamma}\zeta_{u}\overset{S}{=}& 0
\,,
\label{KID_deg11}
\\
 \Delta_{\gamma}\zeta_{v}\overset{S}{=}&
-\theta_S\partial_{[u}\zeta_{v]}
-\frac{\theta_S^2}{n-1}\zeta_{u}
-|{}^{(2)}\pi|^2 \zeta_{u}
\,,
\label{KID_deg12}
\\
\partial_A\partial_{[u}\zeta_{v]}  \overset{S}{=}& -\frac{\theta_S}{n-1} \partial_{A}\zeta_{u}
-{}^{(2)}\pi_{A}{}^B\partial_B\zeta_u
\,.
\label{KID_deg13}
\end{align}

Given data $\zeta|_S$, $\partial_{(u}\zeta_{v)}|_S$ and $\partial_{[u}\zeta_{v]}|_S$
the equations \eq{KID_deg1}-\eq{KID_deg3} and \eq{KID_deg5}-\eq{KID_deg7} can be integrated.
The data $\partial_{(u}\zeta_{v)}|_S$ and $\partial_{[u}\zeta_{v]}|_S$ are  determined by \eq{KID_deg9} and \eq{KID_deg13},
the latter one up to  an additive constant. 
We further observe that  the second-order $u$-derivative of \eq{KID_deg4} vanishes whence it simplifies to
\begin{equation}
(D_{(A}\zeta_{B)} )_{\mathrm{tf}}\overset{S}{=}-{}^{(2)}\pi_{AB}\zeta_{u}
\,,\quad
(D_{A}D_B\zeta_{u})_{\mathrm{tf}}\overset{S}{=} 0
\,,
\end{equation}
and the first equation is contained in \eq{KID_deg8}.

As  the main problematic equation remains  \eq{KID_deg8}.
We do not want to impose  a condition on ${}^{(2)}\pi_{AB}$ which restricts its $v$-dependence.
We rather  want to make sure that $\zeta^v -\frac{1}{2}g^{vv}\zeta_v\overset{N_2^+}{=}0$  is a   solution of \eq{KID_deg7}
to get rid of  ${}^{(2)}\pi_{AB}$ in \eq{KID_deg8}.
Since we want $\zeta_v|_S\ne 0$ in order to get a non-trivial Killing vector we are led to the condition
$\Lambda=0$.
By  \eq{restriction_data2} this requires   $\widehat R_{AB}\overset{N_2^+}{=}0$, and the latter one holds
if $\partial_C{}^{(2)}\pi_{AB}=0$ and $\partial_C g_{AB}\overset{S}{=}0$, as then $g_{AB}|_{N_2^+}$ is independent of $x^C$. This suggests to  assume that $(S, \gamma)$ is the  Euclidean space.
 Proposition~\ref{prop_KID_equations2}  now follows immediately
(note that the conditions  on the initial data there  imply \eq{restriction_data1}-\eq{restriction_data2})

%

\subsection{Proof of Proposition~\ref{prop_KID_equations4}}
\label{app_number_killings}

Let us determine how many independent Killing vectors the spacetimes generated in Proposition~\ref{prop_KID_equations2} have.
Set $\mathfrak{f} :=-  \big(\zeta^v -\frac{1}{2}g^{vv}\zeta_v\big)$
 and note that $\zeta_u|_S =-\mathfrak{f} $.
For data as considered in Proposition~\ref{prop_KID_equations3} the KID equations read for data $(\zeta|_S, \alpha_{\zeta}=\mathrm{const.})$ for the Killing vector,%
%
\begin{eqnarray}
\partial_v\zeta_v &\overset{N_2^+}{=}&0
\,,
\label{deg_KIDs1}
\\
\Big(\partial_{v} - \frac{2}{n-1}{}^{(2)}\theta\Big)\zeta_A
-2{}^{(2)}\pi_{A}{}^B\zeta_B  &\overset{N_2^+}{=}& -\partial_{A} \zeta_{v}
\,,
\\
\partial_v\partial_v \mathfrak{f}  &\overset{N_2^+}{=}&0
\,,
\\
(\partial_{(A}\zeta_{B)} )_{\mathrm{tf}}
 &\overset{N_2^+}{=}&
{}^{(2)}\pi_{AB}\mathfrak{f}
\,,
\\
\partial_{A}\partial_B\mathfrak{f} &\overset{S}{=}& 0
\,,
\\
D^A\zeta_A &\overset{S}{=}& \theta_S\mathfrak{f}
\,,
\\
 \Delta_{\gamma}\zeta_{v} &\overset{S}{=}&
-\theta_S\widehat\kappa_{\zeta}
\,,
\\
 \partial_v \mathfrak{f} &\overset{S}{=}& \frac{\theta_S}{n-1}\mathfrak{f}  + \alpha_{\zeta}
\,.
\label{deg_KIDs8}
\end{eqnarray}
These equations are supplemented by an algebraic equation for $\partial_u\zeta_v|_S$ and the
equations  \eq{KID_deg1}-\eq{KID_deg3}
for the evolution on $\mathcal{H}^+_1$, which, given the data on $S$, always admits unique solutions.

Equivalently,
\begin{eqnarray}
\zeta_v&\overset{N_2^+}{=}&\overset{(0)}{\zeta}_v(x^A)
\,,
\\
\Big(\partial_{v} - \frac{2}{n-1}{}^{(2)}\theta\Big)\zeta_A
-2{}^{(2)}\pi_{A}{}^B\zeta_B&\overset{N_2^+}{=}&  -\partial_{A} \overset{(0)}{\zeta}_v
\,,
\label{gen_KID2}
\\
\mathfrak{f} &\overset{N_2^+}{=}& c+c_A x^A +\Big(\frac{\theta_S}{n-1}\big(c+c_A x^A \big)+ \alpha_{\zeta}\Big)v
\,,
\\
(\partial_{(A}\zeta_{B)} )_{\mathrm{tf}}
&\overset{N_2^+}{=}&
{}^{(2)}\pi_{AB}\mathfrak{f}
\,,
\label{gen_KID4}
\\
D^A\zeta_A &\overset{S}{=}&\theta_S(c+c_A x^A)
\,,
\label{gen_KID4B}
\\
 \Delta_{\gamma}\overset{(0)}{\zeta}_v &\overset{S}{=}&
-\theta_S\alpha_{\zeta}
\,.
\label{gen_KID5}
\end{eqnarray}
where $c$ and the $c_A$'s are constants.
The $v$-derivative of \eq{gen_KID4} yields with \eq{gen_KID2} and \eq{gen_KID5}
\begin{equation}
\partial_{A}\partial_B\overset{(0)}{\zeta}_v
\overset{S}{=}
\frac{\theta_S}{n-1}\alpha_{\zeta}\gamma_{AB} -(c+c_C x^C)\partial_v{}^{(2)}\pi_{AB}
\,.
\end{equation}
Taking the $x^C$-derivative yields with   \eq{gen_KID5}
\begin{equation}
c^B\partial_v{}^{(2)}\pi_{AB}
\overset{S}{=}0
\,,\quad
c_{[A}\partial_v{}^{(2)}\pi_{B]C}
\overset{S}{=}0
\,,
\end{equation}
i.e
\begin{equation}
0\overset{S}{=}2\partial_v {}^{(2)}\pi^{AB}\partial_v {}^{(2)}\pi_{A[B}c_{C]}\overset{S}{=}
| \partial_v{}^{(2)}\pi|^2c_C
- \partial_v {}^{(2)}\pi_{AC} \underbrace{\partial_v{}^{(2)}\pi^{AB}c_B}_{=0}
\,.
\end{equation}
%
As, by assumption, $| \partial_v{}^{(2)}\pi|^2|_S\ne 0$ we necessarily have
\begin{equation}
 c_A=0
\,.
\end{equation}
By \eq{gen_KID4} that yields
\begin{equation}
\partial_C(\partial_{(A}\zeta_{B)} )_{\mathrm{tf}} \overset{N_2^+}{=}0
\,.
\label{deriv_tf_relation}
\end{equation}
Next, we evaluate the $x^B$-derivative of \eq{gen_KID2}   on $S$ and take its symmetric trace-free part.
Using \eq{gen_KID4}, we obtain
\begin{equation}
\partial_{v} (\partial_{(A}\zeta_{B)})_{\mathrm{tf}}
\overset{S}{=}  -(\partial_{A}\partial_B\overset{(0)}{\zeta}_v )_{\mathrm{tf}}
\,.
\end{equation}
Combining both equations and using \eq{gen_KID5} we obtain
\begin{equation}
\zeta_{v}\overset{N_2^+}{=} d+ d_Ax^A+ d_{AB} x^Ax^B
\,,
\quad \gamma^{AB}d_{AB}=-\theta_S\alpha_{\zeta}
\,, \quad d_{[AB]}=0
\,,
\end{equation}
in particular,
\begin{equation}
\widehat\Delta  \zeta_v\overset{N_2^+}{=} g^{AB}d_{AB}\overset{N_2^+}{=} -\theta_S\alpha_{\zeta} e^{-\frac{2}{n-1}\int{}^{(2)}\theta\mathrm{d}v}
\,.
\end{equation}
The divergence of \eq{gen_KID2} then gives
\begin{equation}
\partial_{v}(g^{AB} \partial_{A}\zeta_B)
\overset{N_2^+}{=}   -\widehat\Delta \overset{(0)}{\zeta}
\overset{N_2^+}{=}   \theta_S\alpha_{\zeta} e^{-\frac{2}{n-1}\int{}^{(2)}\theta\mathrm{d}v}
\,,
\label{v_deriv_div}
\end{equation}
which we integrate with initial data given by
 \eq{gen_KID4B},
\begin{equation}
g^{AB} \partial_{A}\zeta_B
\overset{N_2^+}{=}   \theta_S\Big( c +\alpha_{\zeta}\int_0^v e^{-\frac{2}{n-1}\int{}^{(2)}\theta\mathrm{d}v}\mathrm{d}v\Big)
\,.
\label{divergence}
\end{equation}
As the right-hand side does not depend on the $x^A$-coordinates it follows with  \eq{deriv_tf_relation}
that
\begin{equation}
\partial_C\partial_{(A}\zeta_{B)}\overset{N_2^+}{=}0
\,.
\end{equation}
By taking the anti-symmetric part w.r.t.\ $(AC)$ we  find that also the $x^C$-derivatives of $\partial_{[A}\zeta_{B]}$
need to vanish whence
\begin{equation}
\partial_A\partial_{B}\zeta_{C} \overset{N_2^+}{=}0
\quad \Longrightarrow
\quad
\zeta_{A} \overset{N_2^+}{=}\overset{(0)}\zeta_A(v)+ \overset{(1)}\zeta_{AB}(v)x^B
\,,
\end{equation}
where, by \eq{gen_KID4} and \eq{divergence},
\begin{eqnarray}
 \overset{(1)}\zeta_{(AB)}(v) \overset{N_2^+}{=}  \frac{\theta_S}{n-1}\Big( c +\alpha_{\zeta}\int_0^v e^{-\frac{2}{n-1}\int{}^{(2)}\theta\mathrm{d}v}\mathrm{d}v\Big)g_{AB}
+\Big[ c+\Big(\frac{\theta_S}{n-1}c+ \alpha_{\zeta}\Big)v\Big]{}^{(2)}\pi_{AB}
\,.
\label{eqn_coeff}
\end{eqnarray}
We need to make sure that \eq{gen_KID2}  is fulfilled.
 If we plug in the corresponding expressions we have found for $\zeta_A$ and $\zeta_v$ we find that \eq{gen_KID2}  is equivalent to the
following ODE system,
\begin{eqnarray}
\Big(\partial_{v} - \frac{2}{n-1}{}^{(2)}\theta\Big)\overset{(0)}\zeta_A(v)
-2{}^{(2)}\pi_{A}{}^B\overset{(0)}\zeta_B(v)& \overset{N_2^+}{=}&  -  d_A
\,,
\label{ODE_zetaA0}
\\
\Big(\partial_{v} - \frac{2}{n-1}{}^{(2)}\theta\Big) \overset{(1)}\zeta_{AB}(v)
-2{}^{(2)}\pi_{A}{}^C \overset{(1)}\zeta_{CB}(v)& \overset{N_2^+}{=}&  - d_{AB}
\,.
\end{eqnarray}
%
%
Because of \eq{eqn_coeff}  it is useful to split the latter equation into symmetric and anti-symmetric part,%
\begin{align}
\Big(\partial_{v} - \frac{2}{n-1}{}^{(2)}\theta\Big) \overset{(1)}\zeta_{[AB]}(v)
-{}^{(2)}\pi_{A}{}^C \overset{(1)}\zeta_{[CB]}(v)
+ {}^{(2)}\pi_{B}{}^C \overset{(1)}\zeta_{[CA]}(v)
\nonumber
&&
\\
-{}^{(2)}\pi_{A}{}^C \overset{(1)}\zeta_{(BC)}(v)
+ {}^{(2)}\pi_{B}{}^C \overset{(1)}\zeta_{(AC)}(v)
  \overset{N_2^+}{=}&0
\,,
\label{ev_anti_sym}
\\
\Big(\partial_{v} - \frac{2}{n-1}{}^{(2)}\theta\Big) \overset{(1)}\zeta_{(AB)}(v)
-{}^{(2)}\pi_{A}{}^C \overset{(1)}\zeta_{(CB)}(v)
-{}^{(2)}\pi_{B}{}^C \overset{(1)}\zeta_{(CA)}(v)
\nonumber
&
\\
-{}^{(2)}\pi_{A}{}^C \overset{(1)}\zeta_{[CB]}(v)
-{}^{(2)}\pi_{B}{}^C \overset{(1)}\zeta_{[CA]}(v)
 \overset{N_2^+}{=}&  - d_{AB}
\,.
\label{ev_sym}
\end{align}
We observe that the last equation requires
\begin{align}
d_{AB}=& -\partial_{v} \overset{(1)}\zeta_{(AB)}(0)+ \frac{2}{n-1}\theta_S \overset{(1)}\zeta_{(AB)}(0)
\,,
\\
0 =& \partial_v\partial_{v} \overset{(1)}\zeta_{(AB)}(0)+ \frac{2}{(n-1)^2}(\theta_S)^2 \overset{(1)}\zeta_{(AB)}(0)
- \frac{2}{n-1}\theta_S\partial_v \overset{(1)}\zeta_{(AB)}(0)
\nonumber
\\
&
-\partial_v{}^{(2)}\pi_{A}{}^C \overset{(1)}\zeta_{(CB)}(0)
-\partial_v{}^{(2)}\pi_{B}{}^C \overset{(1)}\zeta_{(CA)}(0)
\nonumber
\\
&
-\partial_v{}^{(2)}\pi_{A}{}^C \overset{(1)}\zeta_{[CB]}(0)
-\partial_v{}^{(2)}\pi_{B}{}^C \overset{(1)}\zeta_{[CA]}(0)
\,.
\end{align}
As we have
\begin{align}
 \overset{(1)}\zeta_{(AB)}(0)=&c  \frac{\theta_S}{n-1}\gamma_{AB}
\,,
\\
\partial_v \overset{(1)}\zeta_{(AB)}(0) =&
\alpha_{\zeta} \frac{\theta_S}{n-1} \gamma_{AB}
+2c\Big( \frac{\theta_S}{n-1}\Big)^2 \gamma_{AB}
+c \partial_v{}^{(2)}\pi_{AB}
\,,
\\
\partial_v \partial_v  \overset{(1)}\zeta_{(AB)}(0)=&
 2\Big(\frac{\theta_S}{n-1}\Big)^2\alpha_{\zeta}\gamma_{AB}
+2c\Big(\frac{\theta_S}{n-1}\Big)^3\gamma_{AB}
\nonumber
\\
&
+2\Big(2c\frac{\theta_S}{n-1}+ \alpha_{\zeta}\Big)\partial_v{}^{(2)}\pi_{AB}
+ c\partial_v\partial_v{}^{(2)}\pi_{AB}
\,,
\end{align}
that yields
\begin{align}
d_{AB} \overset{S}{=}& -\alpha_{\zeta} \frac{\theta_S}{n-1} \gamma_{AB}
-c \partial_v{}^{(2)}\pi_{AB}
\,,
\\
2 \alpha_{\zeta}\partial_v{}^{(2)}\pi_{AB}\overset{S}{=}&
\partial_v{}^{(2)}\pi_{A}{}^D \overset{(1)}\zeta_{[DB]}(0)
+\partial_v{}^{(2)}\pi_{B}{}^D \overset{(1)}\zeta_{[DA]}(0)
-c\partial_v\partial_v{}^{(2)}\pi_{AB}
\,.
\label{equation_kappa}
\end{align}
To conclude, the data
\begin{eqnarray}
\zeta_{v}&\overset{S}{=}& d+ d_Ax^A - \Big(\alpha_{\zeta} \frac{\theta_S}{n-1} \gamma_{AB}
+c \partial_v{}^{(2)}\pi_{AB}\Big) x^Ax^B
\,,
\\
\zeta_A&\overset{S}{=}& \overset{(0)}\zeta_A+\Big( \overset{(1)}\zeta_{[AB]}+c  \frac{\theta_S}{n-1}\gamma_{AB}\Big)x^B
\\
\mathfrak{f}&\overset{N_2^+}{=} &c+\Big(c\frac{\theta_S}{n-1}+  \alpha_{\zeta}\Big)v
\,,
\end{eqnarray}
where $d$, $d_A$,  $\overset{(0)}\zeta_A|_S$, $\overset{(1)}\zeta_{[AB]}|_S$ and $c$ are constant on $S$, and where
 $ \alpha_{\zeta} $  is determined  by \eq{equation_kappa},
 exhaust all candidates to generate  a unique solution to  \eq{deg_KIDs1}-\eq{deg_KIDs8}.
They do whenever  \eq{ev_sym} holds  with $\overset{(1)}\zeta_{[AB]}(v)$ determined by \eq{ev_anti_sym}
and $\overset{(1)}\zeta_{(AB)}(v)$ as given by \eq{eqn_coeff}.
In that case $\overset{(0)}\zeta_A(v)$ is determined by \eq{ODE_zetaA0}.

This will certainly be the case for $c=\overset{(1)}\zeta_{[AB]}(0)=0$ (which implies $\alpha_{\zeta}=0$),
whence
we have at least $2n-1$ independent  solutions which are parameterized by $d$, $d_A$ and $ \overset{(0)}\zeta_A(0)$,
while for specific choices of the data ${}^{(2)}\pi_{AB}$ and  $\theta_S$  there might be $\frac{1}{2}(n-1)(n-2)+1$
 additional Killing vectors. This completes the proof of Proposition~\ref{prop_KID_equations4}.

\subsection{Proof of Proposition~\ref{prop_KID_equations5}}
\label{app_kids_deg_van_shear}

First of all we observe that the constraint equations \eq{constr1}-\eq{constr4} yield
(recall our assumptions \eq{ass_vanishing_shear1}-\eq{ass_vanishing_shear2} on the initial data)
%
\begin{eqnarray*}
{}^{(1)}\theta&\overset{\mathcal{H}^+_1}{=}& 0
\,,
\\
\Gamma^u_{uA}  &\overset{\mathcal{H}^+_1}{=}&-\varsigma_A(x^B)
\,,
\\
{}^{(2)}\theta &\overset{N_2^+}{=}& \Big(\frac{v}{n-1} +\frac{1}{\theta_S}\Big)^{-1}
\,,
\\
g_{AB} &\overset{N_2^+}{=}&  
 \Big(1 + \frac{\theta_S}{n-1} v\Big)^2 \gamma_{AB}
\,,
\\
\Gamma^v_{vA} &\overset{N_2^+}{=}&
\frac{n-1 }{\theta_S(n-1+\theta_S  v)}\partial_A\theta_S
\,.
\end{eqnarray*}
Taking the behavior of the Ricci tensor under conformal transformations into account  that yields
%
\begin{eqnarray*}
(\widehat R_{AB})_{\mathrm{tf}}&\overset{N_2^+}{=}&  ({}^{\gamma} R_{AB})_{\mathrm{tf}} -\frac{(n-3)v }{n-1+ \theta_S v}(D_AD_B\theta_S)_{\mathrm{tf}}
\\
&&
+2 (n-3)\Big(\frac{v}{n-1 + \theta_S v}\Big)^2(D_A \theta_SD_B \theta_S)_{\mathrm{tf}}
\,,
\\
\widehat R &\overset{N_2^+}{=}& \Big(\frac{n-1}{n-1 + \theta_S v}\Big)^2 \Big({}^{\gamma}  R -\frac{2(n-2)v }{n-1+ \theta_S v}\Delta_{\gamma}\theta_S
\\
&&
-(n-2)(n-5)\Big(\frac{v}{n-1 + \theta_S v}\Big)^2D_A \theta_SD^A \theta_S\Big)
\,.
\end{eqnarray*}
We then deduce that $^{(2)}\Xi_{AB}$ is determined by the following equations, with trivial initial data on $S$,
\begin{eqnarray*}
\Big(\partial_v+  \Big(\frac{v}{n-1} +\frac{1}{\theta_S}\Big)^{-1}\Big)\mathrm{tr}({}^{(2)}\Xi)
 - 2(n-1)^2 \frac{n-1+ (n-2)v\theta_S}{(n-1 + \theta_S v)^3} \frac{\Delta_{\gamma}\theta_S}{\theta_S}
&&
\\
-(n-1)^2(n-5)v\frac{2(n-1) + (n-2)v\theta_S}{(n-1 + \theta_S v)^4} \frac{ |D \theta_S|^2}{\theta_S}
+\frac{(n-1)^2}{(n-1 + \theta_S v)^2} {}^{\gamma}  R
&\overset{N_2^+}{=}& 2\Lambda
\,,
\\
\Big(\partial_{v}+\frac{n-5}{2(n-1)}\Big(\frac{v}{n-1} +\frac{1}{\theta_S}\Big)^{-1} \Big)({}^{(2)}\Xi_{AB})_{\mathrm{tr}}
-\frac{2(n-1 )+(n-3)v\theta_S}{\theta_S(n-1+\theta_S  v)}( D_{A} D_{B}\theta_S)_{\mathrm{tf}}\hspace{-5em}
&&
\\
+\frac{8(n-1)   v+ 2 (n-3)v^2\theta_S}{\theta_S(n-1+\theta_S  v)^2}( D_{A}\theta_SD_{B}\theta_S)_{\mathrm{tf}}
+ ({}^{\gamma}  R_{AB})_{\mathrm{tf}}
&\overset{N_2^+}{=}&0
\,.
\end{eqnarray*}

%
%

Next, we evaluate the KID equations \eq{KID_deg_hor1}-\eq{KID_deg_hor10} in this setting,
omitting the ---for our purposes  irrelevant--- 2nd-order ODEs  for $\zeta^u|_{\mathcal{H}^+_1}$ and  $\zeta^v|_{N^+_2}$,
%
\begin{eqnarray}
\zeta_u  &\overset{\mathcal{H}^+_1}{=}& 0
\,,
\quad
 \zeta_{A}\,\overset{\mathcal{H}^+_1}{=}\,0
\,,
\label{KID_non_bifurcate1}
\\
\zeta_v  &\overset{N_2^+}{=}& \overset{(0)}{\zeta }(x^A)
\,,
\\
\zeta_A&\overset{N_2^+}{=}&
 vD_A \overset{(0)}{\zeta } -\theta_S v\frac{ 2(n-1)  + \theta_S v}{n-1}D_A(\theta_S ^{-1}\overset{(0)}{\zeta })
\,,
\\
0 &\overset{N_2^+}{=}& \Big(D_{(A}\zeta_{B)} -\frac{1}{2}{}^{(2)}\Xi_{AB}\overset{(0)}{\zeta }- \frac{2v}{n-1} \Big(1 + \frac{\theta_S}{n-1} v\Big)^{-1} \zeta_{(A}D_{B)}\theta_S \Big)_{\mathrm{tf}}
\,,
\\
\partial_{(u}\zeta_{v)} &\overset{S}{=}&\partial_u g_{uv}\overset{(0)}{\zeta}
\,,
\label{KID_non_bifurcate8}
\\
\Delta_{\gamma}(\theta_S^{-1} \overset{(0)}{\zeta })
+ \kappa
&\overset{S}{=}&  \frac{1}{2 }\Big({}^{\gamma} R - 2\Lambda\Big)(\theta_S^{-1}\overset{(0)}{\zeta })
\,,
\label{KID_non_bifurcate9}
\\
\partial_{[u}\zeta_{v]}   &\overset{S}{=}&\kappa
\,, \quad \kappa=\mathrm{const.}
\label{KID_non_bifurcate10}
\end{eqnarray}
%
%
In fact, we observe that the above system admits a solution if and only if there exists a function
$\overset{(0)}{\zeta }$ on $S$ and a constant $\kappa$ such that
\begin{eqnarray}
0 &\overset{N_2^+}{=}& \Big(D_{(A}\zeta_{B)} -\frac{1}{2}{}^{(2)}\Xi_{AB}\overset{(0)}{\zeta }- \frac{2v}{n-1} \Big(1 + \frac{\theta_S}{n-1} v\Big)^{-1} \zeta_{(A}D_{B)}\theta_S \Big)_{\mathrm{tf}}
\,,
\label{crucial_eqn}
\\
\Delta_{\gamma}(\theta_S^{-1} \overset{(0)}{\zeta })
+ \kappa
&\overset{S}{=}&  \frac{1}{2 }\Big({}^{\gamma} R - 2\Lambda\Big)(\theta_S^{-1}\overset{(0)}{\zeta })
\,,
\end{eqnarray}
where
\begin{eqnarray}
\zeta_A&\overset{N_2^+}{=}&
 vD_A \overset{(0)}{\zeta} -\theta_S v\frac{ 2(n-1)  + \theta_S v}{n-1}D_A(\theta_S ^{-1}\overset{(0)}{\zeta })
\,.
\label{expr_zeta_A}
\end{eqnarray}
The main obstrution comes from \eq{crucial_eqn}.
Note that it is automatically satisfied on $S$.
We want to compute its $v$-derivative on $S$. For this note that it follows from the constraint equations
and \eq{expr_zeta_A}
that
\begin{eqnarray*}
\partial_{v}  \zeta_{A} &\overset{S}{=}&
- D_{A}\overset{(0)}{\zeta} +  \frac{2 }{\theta_S}\overset{(0)}{\zeta}D_A\theta_S
\,,
\\
\partial_{v}({}^{(2)}\Xi_{AB})_{\mathrm{tr}}
&\overset{S}{=}&
\frac{2}{\theta_S}( D_{A} D_{B}\theta_S)_{\mathrm{tf}}
-({}^{\gamma} R_{AB})_{\mathrm{tf}}
\,,
\end{eqnarray*}
whence we find for the $v$-derivative of \eq{crucial_eqn} on $S$,
\begin{equation}
\Big(D_{A}D_{B}(\theta_S^{-1}\overset{(0)}{\zeta })
\Big)_{\mathrm{tf}}\overset{S}{=} \frac{1}{2} (\theta_S^{-1}\overset{(0)}{\zeta })({}^{\gamma} R_{AB})_{\mathrm{tf}}
\,.
\label{cond_KV}
\end{equation}
Similarly, from
\begin{eqnarray*}
\partial_{v}^2 \zeta_{A}    &\overset{S}{=}&
- \frac{2}{n-1}\theta_S ^2D_A(\theta_S ^{-1}\overset{(0)}{\zeta })
\\
\partial_{v}^2({}^{(2)}\Xi_{AB})_{\mathrm{tr}}
&\overset{S}{=}&
-\frac{8   }{\theta_S(n-1)}( D_{A}\theta_SD_{B}\theta_S)_{\mathrm{tf}}
+\frac{n-5}{2(n-1)} \theta_S({}^{\gamma} R_{AB})_{\mathrm{tf}}
\,,
\end{eqnarray*}
we find with   \eq{cond_KV}  the equation $\overset{(0)}{\zeta }({}^{\gamma} R_{AB})_{\mathrm{tf}} =0$, for the 2nd-order $v$-derivative  of \eq{crucial_eqn}. As we are interested in MKHs we require the existence of a solution
where $\overset{(0)}{\zeta}\ne 0$ whence this condition becomes
\begin{equation}
({}^{\gamma} R_{AB})_{\mathrm{tf}} =0
\,.
\label{cond_KV2}
\end{equation}
Using this, we find from
\begin{eqnarray*}
\partial_{v}^3 \zeta_{A}    &\overset{S}{=}&
0
\,,
\\
\partial_{v}^3({}^{(2)}\Xi_{AB})_{\mathrm{tr}}
&\overset{S}{=}&
\frac{24  }{(n-1)^2}( D_{A}\theta_SD_{B}\theta_S)_{\mathrm{tf}}
\,,
\end{eqnarray*}
that  the 3rd-order $v$-derivative of \eq{crucial_eqn} yields the following condition
\begin{equation}
\Big(  D_{(A}\theta_SD_{B)}(\theta_S ^{-1}\overset{(0)}{\zeta }) \Big)_{\mathrm{tf}}\overset{S}{=} 0
\,.
\label{cond_KV3}
\end{equation}
Let us analyze to what extent the conditions obtained so far are already sufficient, i.e.\ we consider again \eq{crucial_eqn}.
We observe that with \eq{cond_KV2} the solution to the $({}^{(2)}\Xi_{AB})_{\mathrm{tf}}$-constraint reads
\begin{equation*}
(^{(2)}\Xi_{AB})_{\mathrm{tf}} \overset{N_2^+}{=}\Big(\frac{2 v}{\theta_S} D_{A}D_{B}\theta_S
-\frac{4v^2}{\theta_S(n-1+\theta_S v)}D_{A}\theta_SD_{B}\theta_S\Big)_{\mathrm{tf}}
\,.
\end{equation*}
Inserting this and \eq{expr_zeta_A} into \eq{crucial_eqn} and using \eq{cond_KV}-\eq{cond_KV3} we find
that \eq{crucial_eqn} is fulfilled on  $N_2^+$.
This accomplishes the proof of Proposition~\ref{prop_KID_equations5} (where we have set  $f:=\theta_S^{-1}\overset{(0)}{\zeta }$).

\end{document}